\documentclass[10pt]{article}
\usepackage{epsfig,amsmath,graphicx,amssymb,overpic}

\usepackage{graphicx}
\usepackage{dcolumn}
\usepackage{bm}
\usepackage{graphicx}
\usepackage{subfigure}
\usepackage{epsfig,amsmath,graphicx,amssymb,overpic,cite}

\usepackage{graphicx}
\usepackage{dcolumn}
\usepackage{bm}
\usepackage{graphicx}
\usepackage{subfigure}
\usepackage{amsthm}

\newtheorem{corollary}{Corollary}
\newtheorem{lemma}{Lemma}
\newtheorem{proposition}{Proposition}
\newtheorem{remark}{Remark}
\newtheorem{theorem}{Theorem}

\def\be{\begin{equation}}
\def\ee{\end{equation}}
\def\bee{\begin{eqnarray}}
\def\ene{\end{eqnarray}}
\def\bes{\begin{subequations}}
\def\ees{\end{subequations}}

\def\d{\displaystyle}
\def\v{\vspace{0.1in}}
\def\no{{\nonumber}}

\begin{document}

\baselineskip=13pt
\renewcommand {\thefootnote}{\dag}
\renewcommand {\thefootnote}{\ddag}
\renewcommand {\thefootnote}{ }

\pagestyle{plain}

\begin{center}
\baselineskip=16pt \leftline{} \vspace{-.3in} {\Large \bf Inverse Scattering Transforms for the Focusing and Defocusing MKdV Equations with Nonzero Boundary Conditions} \\[0.2in]
\end{center}

\begin{center}
Guoqiang Zhang$^{1,2}$ and Zhenya Yan$^{1,2,*}$\footnote{$^{*}${\it Email address}: zyyan@mmrc.iss.ac.cn}  \\[0.03in]
{\it \small $^{1}$Key Laboratory of Mathematics Mechanization, Academy of Mathematics and Systems Science, \\ Chinese Academy of Sciences, Beijing 100190, China \\
 $^{2}$School of Mathematical Sciences, University of Chinese Academy of Sciences, Beijing 100049, China} \\
\end{center}

\vspace{0.5in}

{\baselineskip=14pt


\vspace{-0.18in}

\noindent {\bf Abstract:}\, We explore systematically a rigorous theory of the inverse scattering transforms with matrix Riemann-Hilbert problems for both focusing and defocusing modified Korteweg-de Vries (mKdV) equations with non-zero boundary conditions (NZBCs) at infinity. Using a suitable uniformization variable, the direct and inverse scattering problems are proposed on a complex plane instead of a two-sheeted Riemann surface. For the direct scattering problem, the analyticities, symmetries and asymptotic behaviors of the Jost solutions and  scattering matrix, and discrete spectra are established. The inverse problems are formulated and solved with the aid of the Riemann-Hilbert problems, and the reconstruction formulae, trace formulae, and theta conditions are also found. In particular, we present the general solutions for the focusing mKdV equation with NZBCs and both simple and double poles, and for the defocusing mKdV equation with NZBCs and simple poles. Finally, some representative reflectionless potentials are in detail studied to illustrate distinct wave structures for both focusing and defocusing mKdV equations with NZBCs.





\vspace{0.15in}

\section{Introduction}

We consider the inverse scattering transforms and general solutions for both focusing and defocusing mKdV equations with non-zero boundary conditions (NZBCs) at infinity, namely
\bee\label{mKdV}
\begin{array}{l}
 q_t-6\,\sigma q^2q_x+q_{xxx}=0, \quad (x,t)\in \mathbb{R}^2,\v\\
\displaystyle\lim_{x\to\pm\infty}q(x,t)=q_{\pm},\quad \left|q_{\pm}\right|=q_0\ne0
\end{array}
\ene
based on the generalized matrix Riemann-Hilbert problem, not the usual Gel'fand-Leviton-Machenko integral equation~\cite{ist1,ist2}, where $q=q(x, t)\in\mathbb{R}$, $\sigma=-1, 1$ denote, respectively, the focusing and defocusing mKdV. Eq.~(\ref{mKdV}) possesses the scaling symmetry $q(x,t)\to \alpha q(\alpha x, \alpha^3t)$ with $\alpha$ being a non-zero parameter, and can be rewritten as a Hamiltonian system under the non-zero condition
\bee
 q_t=\frac{\partial}{\partial x}\frac{\delta H[q]}{\delta q},\qquad H[q]=\frac12\int_{-\infty}^{\infty}\left(\sigma q^4+q_x^2-\sigma q_0^4\right)\,\mathrm{d}x.
\ene
Eq.~(\ref{mKdV}) arises in many different physical contexts, such as acoustic wave and phonons in a certain anharmonic lattice \cite{Zabusky1967, Ono1992}, Alfv\'{e}n wave in a cold collision-free plasma \cite{Kakutani1969, Khater1998}, thin elastic rods \cite{Matsutani1991}, meandering ocean currents \cite{Ralph1994}, dynamics of traffic flow \cite{Komatsu1995, Ge2005}, hyperbolic surfaces \cite{Schief1995}, slag-metallic bath interfaces \cite{Agop1998}, and Schottky barrier transmission lines \cite{Ziegler2001}. The well-known Miura transform~\cite{ist2, miura} $u(x,t)=-\sigma q^2(x,t)+\sqrt{\sigma} q_x(x,t)$ established the relation between Eq.~(\ref{mKdV}) and the KdV equation $u_t+6uu_x+u_{xxx}=0$. The Hodograph transform~\cite{ht,ht2} $r(y,t)=e^{1/2\int^x q(x,t)dx},\, x=\int^y u^{-1}(y,t)dy$ was found between Eq.~(\ref{mKdV}) and the Harry-Dym equation $u_t=u^3u_{yyy}$. Moreover, there exists a transformation~\cite{Alejo18} $p(x,t)=q(x-6\sigma \nu^2t, t)-\nu$ between  Eq.~(\ref{mKdV}) equation and the Gardner equation (also called the combined KdV-mKdV equation)~\cite{ist2, geq} with BCs
\bee\label{geq}
\begin{array}{l}
 p_t -6\,\sigma (p^2+2\nu p)p_x+p_{xxx}=0, \quad \nu\in \mathbb{R}\backslash \{0\}, \,\, (x,t)\in \mathbb{R}^2, \v\\
\displaystyle\lim_{x\to\pm\infty}p(x,t)=p_{\pm},\quad p_{\pm}=q_{\pm}-\nu
\end{array}
\ene

Since the inverse scattering transform (IST) method was presented to solve the initial-value problem for the integrable KdV equation by Gardner, Greene, Kruskal and Miura~\cite{Gardner1967}, and for the nonlinear Schr\"odinger equation by Zakharov and Shabat~\cite{Shabat1972,Zakharov1973}, there have been some results on the IST for the mKdV equation. For instance, Wadati studied the focusing mKdV equation with zero boundary conditions (ZBCs) and derived simple-pole, double-pole and triple-pole solutions~\cite{Wadati1973, Wadati1982}, after which the $N$-soliton solutions and breather solutions for the focusing mKdV equation with ZBCs were found~\cite{Demontis2011}. Deift and Zhou first presented the long-time asymptotics of the defocusing mKdV equation with ZBCs by using the powerful steepest descent method~\cite{dz1993}. Recently, Germain {\it et al.} studied a full asymptotic stability for solitons of the Cauchy problem for the focusing mKdV equation~\cite{germain16}. The focusing mKdV equation with NZBCs was also studied such that the $N$-soliton solutions were obtained for the simple-pole case with pure imaginary discrete spectra~\cite{Au-Yeung1984, Yeung1988}, and the breather solutions were found for the simple-pole case with a pair of conjugate complex discrete spectra~\cite{Alejo2012}. The Hamiltonian formalism of the defocusing mKdV equation with special NZBCs $q_+=q_-$ was given~\cite{hej}. Recently, the long-time asymptotics of the simple-pole solution for the focusing mKdV equation with step-like NZBCS $q_- \neq q_+=0$ was studied~\cite{bal}.

To the best of our knowledge, there still are the following open questions on the ISTs for the defocusing mKdV equation with NZBCs and even for the focusing mKdV equation with NZBCs:

\begin{itemize}

{\bf\item Though some special simple-pole solutions of the focusing mKdV equation with NZBCs were given~\cite{Au-Yeung1984, Yeung1988,Alejo2012},  there still exists a natural problem whether it admits a general simple-pole solution with mixed pairs of conjugate complex discrete spectra and pure imaginary discrete spectra, i.e., $N_1$-breather-$N_2$-soliton solutions ($N_1+N_2=N$).

\item For the focusing mKdV equation with NZBCs, the multi-pole solutions, i.e., the  solutions corresponding to multiple-pole of the reflection coefficients, were not proposed yet. Especially, the general double-pole solutions with pairs of conjugate complex and pure imaginary discrete spectra were also unknown yet.

\item The inverse problem of the focusing mKdV equation with NZBCs was solved by using the Gel'fand-Leviton-Machenko equation before~\cite{bal}, rather than formulated in terms of a matrix Riemann-Hilbert problem via a suitable uniformization variable.

\item Though there are some partial results on the IST for the focusing mKdV equation with NZBCs, a more rigorous theory of the IST for the  focusing mKdV equation with NZBCs remains open, such as the Riemann surface, uniformaization variable, analyticity, the symmetries and the asymptotic of Jost solutions and scattering matrix, reconstruction formula, trace formula and theta conditions.

\item There are almost no known results on the IST for the defocusing mKdV equation with NZBCs.}
\end{itemize}

 Recently, Ablowitz, Biondini, Demontis, Prinary, {\it et al} presented a powerful approach to study the ISTs for some nonlinear Schr\"{o}dinger (NLS)-type equations with NZBCs at infinity~\cite{Prinari2006, Ablowitz2007, Prinari2011, Demontis2013, Demontis2014, Biondini2014, Biondini2014a, Kraus2015, Prinari2015, Prinari2015a, Mee2015, Biondini2015, Biondini2016, Biondini2016a, Pichler2017, Ablowitz2018, Prinari2018}, in which the inverse problems were formulated in terms of the suitable Riemann-Hilbert problems by defining uniformization variables. Inspired by the above-mentioned idea, in this paper we would like to develop a more general theory to study systematically the ISTs for both focusing and defocusing mKdV equations with NZBCs (\ref{mKdV}) to solve affirmatively those above-mentioned open issues in turn. It should be pointed out that the mKdV equation differs from the NLS equation for two main reasons: i) symmetries of their Lax pairs are different; ii) the potential in the NLS equation is complex while one in the mKdV equation is real such that more complicated conditions are required for the mKdV equation. Moreover, based on the above-mentioned relations between the mKdV equation Eq.~(\ref{mKdV}) and other physically important nonlinear wave equations, the obtained results can also be applied to these nonlinear wave equations.

   The rest of this paper is organized as follows. In Sec. II, we present a rigorous theory of the IST for the focusing mKdV equation with NZBCs and simple poles. Moreover, the inverse problem of the focusing mKdV equation with NZBCs was formulated in terms of a matrix Riemann-Hilbert problem by a suitable uniformization variable. As a results, a general simple-pole solution with pairs of conjugate complex discrete spectra and pure imaginary discrete spectra, i.e., $N_1$-breather-$N_2$-soliton solutions, are found. In Sec. III, we derive the IST for the focusing mKdV equation with NZBCs and double poles such that a general $N_1$-(breather, breather)-$N_2$-(bright, dark)-soliton solutions are found. In Sec. IV, the IST for the defocusing mKdV equation with NZBCs and simple poles is presented to generate multi-dark-soliton-kink solutions for some special cases. Finally, we give the conclusions and discussions in Sec. V.

It is well-known that the focusing or defocusing mKdV equation (\ref{mKdV}) is the compatibility condition, $X_t-T_x+[X, T]=0$, of  the ZS-AKNS
scattering problem (i.e., Lax pair)~\cite{Ablowitz1973}
\begin{alignat}{2}
\label{lax-x}
\varPhi_x&=X\varPhi, &\quad X(x ,t; k)&=ik\sigma_3+Q,  \\ \label{lax-t}
\varPhi_t&=T\varPhi, &\quad T(x, t; k)&=4k^2X-2ik\sigma_3\left(Q_x-Q^2\right)+2Q^3-Q_{xx},
\end{alignat}
where $k$ is an iso-spectral parameter, the eigenfunction $\varPhi(x, t; k)$ is chosen as a $2\times 2$ matrix, the potential matrix $Q$ is given by
\begin{align}
Q=Q(x, t)=
\begin{bmatrix}
0&q(x, t) \\
\sigma q(x, t)&0
\end{bmatrix},
\end{align}
and $\sigma=-1$ and $\sigma=1$ correspond to the focusing and defocusing mKdV equations, respectively.

\begin{remark}	The complex conjugate and conjugate transpose are denoted by $*$ and $\dag$, respectively. The three Pauli matrices are defined as
\begin{align}
\sigma_1=\begin{bmatrix}
0&1\\1&0
\end{bmatrix},\quad
\sigma_2=\begin{bmatrix}
0&-i\\i&0
\end{bmatrix},\quad
\sigma_3=\begin{bmatrix}
1&0\\
0&-1
\end{bmatrix},
\end{align}
and $e^{\alpha\widehat \sigma_3}A:=e^{\alpha\sigma_3}Ae^{-\alpha\sigma_3}$ with $A$ being a $2\times 2$ matrix and $\alpha$ being a scalar variable.
\end{remark}

\section{The focusing mKdV equation with NZBCs: simple poles}

\subsection{Direct scattering problem with NZBCs}

The direct scattering process can determine the analyticity and the asymptotic of the scattering eigenfunctions, symmetries, and asymptotics of the scattering matrix, discrete spectrum, and residue conditions. Though defining a suitable uniformization variable~\cite{Faddeev1987}, the two-sheeted Riemann surface for $k$ can be transformed into the standard complex $z$-plane, on which the scattering problem is  discussed conveniently.

\subsubsection{Riemann surface and uniformization variable}

Considering the following asymptotic scattering problem ($x\to \pm\infty$) of the focusing Lax pair (\ref{lax-x}) and (\ref{lax-t}) with $\sigma=-1$:
\begin{align}\label{alax}
\begin{aligned}
\varPhi_x&=X_{\pm}(k)\varPhi, & X_{\pm}(k)&=\lim_{x\to\pm \infty}X(k)=ik\sigma_3+Q_{\pm}, \\[0.04in]
\varPhi_t&=T_{\pm}(k)\varPhi, &T_{\pm}(k)&=\lim_{x\to\pm \infty}T(k)=\left(4k^2-2q_0^2\right)X_{\pm}(k),
\end{aligned}
\end{align}
with
$Q_{\pm}=\lim\limits_{x\to\pm \infty} Q(x,t)=
\begin{bmatrix}
0&q_{\pm}\\
-q_{\pm}&0
\end{bmatrix},$ one can obtain the fundamental matrix solution of Eq.~(\ref{alax}) as
\begin{align}
\varPhi^{bg}(x, t; k)=
\left\{
\begin{alignedat}{2}
&E_{\pm}(k)\,\mathrm{e}^{i\theta(x, t; k)\sigma_3}, &&k\ne\pm iq_0,\\[0.04in]
&I+\left(x-6\,q_0^2\,t\right)X_{\pm}(k),&\quad&k=\pm iq_0,
\end{alignedat}\right.
\end{align}
where
\begin{align} \label{lp}
E_{\pm}(k)=
\begin{bmatrix}
1&\dfrac{iq_{\pm}}{k+\lambda} \vspace{0.05in} \\
\dfrac{iq_{\pm}}{k+\lambda}&1
\end{bmatrix}=I+\dfrac{iq_{\pm}}{k+\lambda}\sigma_1,\quad
\theta(x, t; k)=\lambda\left[x+\left(4k^2-2q_0^2\right)t\right], \quad \lambda(k)=\sqrt{k^2+q_0^2}.
\end{align}

To discuss the analyticity of the Jost solutions, one needs to determine the regions where Im$\,\lambda(k)>0\,(<0)$ in $\theta(x, t; k)$ (cf. Eq.~(\ref{lp}), see, e.g., Refs.~\cite{Biondini2014,Pichler2017})). Since the defined $\lambda(k)$ is doubly branched, where the branch points are $k=\pm iq_0$, thus one should introduce a two-sheeted Riemann surface such that $\lambda(k)$ is a single-valued function on its each sheet. Let $k\mp iq_0=r_{\pm}\,\mathrm{e}^{i\theta_{\mp}+2im_{\mp}\pi}$ with the arguments $-\pi/2\le\theta_{\mp}<3\pi/2$ and $m_{\mp}\in \mathbb{Z}$, one obtains two single-valued  branches $\lambda_1(k)=\sqrt{r_-r_+}\,\mathrm{e}^{i\frac{\theta_-+\theta_+}{2}}$ and $\lambda_2(k)=-\lambda_1(k)$, respectively, on Sheet-I and Sheet-II, where the branch cut is the segment $i\left[-q_0,q_0\right]$. The region where Im $\lambda(k)>0$ is the upper-half plane (UHP) on Sheet-I and the lower-half plane (LHP) on Sheet-II. The region where Im $\lambda(k)<0$ is the LHP on Sheet I and UHP on Sheet II. Besides, $\lambda(k)$ is real-valued on real $k$ axis and the branch cut (see Fig.~\ref{f}(left)).

\begin{remark} When $q_0=0$, i.e., $q_{\pm}=0$, the NZBCs reduces to the ZBCs and $\lambda=\pm k$.
\end{remark}

\begin{figure}[!t]
\centering
\includegraphics[scale=0.42]{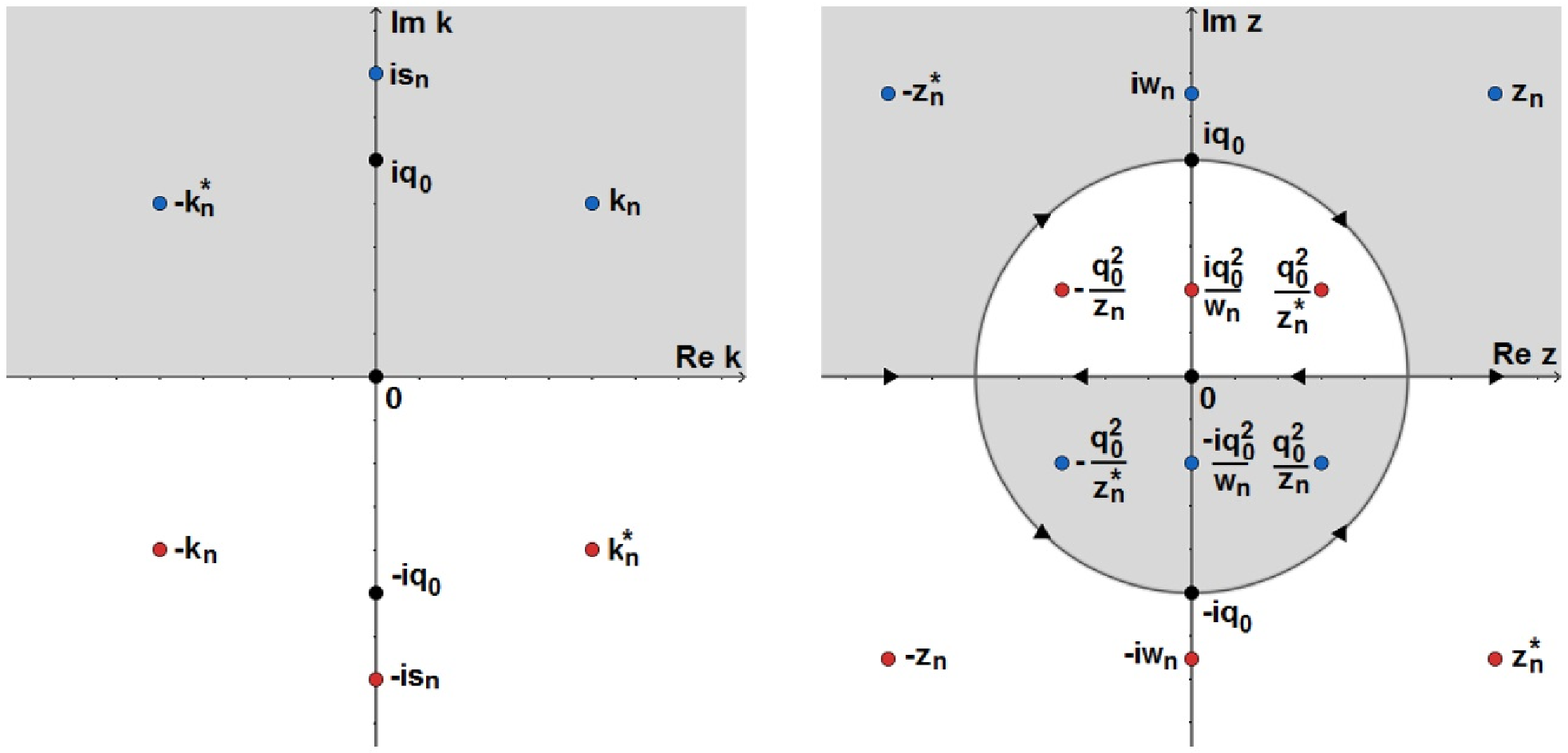}
\caption{Fousing mKdV equation with NZBCs. Left: the first sheet of the Riemann surface, showing the discrete spectrums,  the region where $\mathrm{Im}\, \lambda>0$ (grey) and the region where $\mathrm{Im}\, \lambda<0$ (white). Right: the complex $z$-plane, showing the discrete spectrums [zeros of $s_{11}(z)$ (blue) in grey region and those of $s_{22}(z)$ (red) in white region], the region $D_+$ where  $\mathrm{Im}\,\lambda>0$ (grey), the region $D_-$ where $\mathrm{Im}\,\lambda<0$ (white) and the orientation of the contours for the Riemann-Hilbert problem.}
\label{f}
\end{figure}

Before we continue to study the properties of the Jost solutions and  scattering datas, it is convenient to introduce a uniformization variable $z$  defined by the conformal mapping \cite{Faddeev1987}:
\begin{align}\label{uvz}
z=k+\lambda=k+\sqrt{k^2+q_0^2},
\end{align}
whose inverse mapping is derived as
\begin{align}\label{kz}
k=\frac{1}{2}\left(z-\frac{q_0^2}{z}\right),\quad \lambda=\frac{1}{2}\left(z+\frac{q_0^2}{z}\right).
\end{align}
The mapping relation between the two-sheeted Riemann $k$-surface (Fig.~\ref{f}(left)) and complex $z$-plane (Fig.~\ref{f}(right)) is observed as follows:
\begin{itemize}

\item On the Sheet-I of the Riemann surface, $z\to \infty$ as $k\to \infty$, while on the  Sheet-II of the Riemann surface, $z\to 0$ as $k\to \infty$;

\item The Sheet-I  and Sheet-II, excluding the branch cut, are mapped onto the exterior and interior of the circle of radius $q_0$, respectively;

\item The branch cut $i\left[-q_0,q_0\right]$ is mapped into the circle of radius $q_0$. In particular, the segment $i\left[0,q_0\right]$ of the Sheet-I (Sheet-II) is mapped onto the part in the first (second) quadrant of complex $z$-plane, and the segment $i\left[-q_0,0\right]$ of Sheet-I (Sheet-II) is mapped onto the part in the third (fourth) quadrant of complex $z$-plane;

\item The real $k$ axis is mapped onto the real $z$ axis. In particular, $\left[0,+\infty\right]$ of the Sheet-I (Sheet-II)  is mapped onto $\left[q_0,+\infty\right)$ ($\left[-q_0,0\right]$) and $\left(-\infty, 0\right]$ of the Sheet-I (Sheet-II)  is mapped onto $\left(-\infty,-q_0\right]$ ($\left[0,q_0\right]$);

\item The region for Im $\lambda>0$ (Im $\lambda<0$) of the Riemann surface is mapped onto the grey (white) domain in the complex $z$-plane. In particular, the UHP of the Sheet-I (Sheet-II) is mapped onto the grey (white) domain of the UHP in the complex $z$-plane, and the LHP of the Sheet-I (Sheet-II) is mapped onto the white (grey) domain of the LHP in the complex $z$-plane.
\end{itemize}

For convenience, we denote the grey and white domains in Fig.~\ref{f} (right) by
$ D_+=\{z\in\mathbb{C}: {\rm Im}(\lambda(z))=\frac12
{\rm Im}(z)(1-q_0^2/|z|^2)>0\}$ and
$D_-=\{z\in\mathbb{C}: {\rm Im}(\lambda(z))=\frac12
{\rm Im}(z)(1-q_0^2/|z|^2)<0\},$
 respectively. In the following, we will consider our problem on the complex $z$-plane instead of $k$-plane. With the help of the inverse mapping (\ref{kz}), one can rewrite the fundamental matrix solution of the asymptotic scattering problem as $\varPhi(x, t; z)=E_{\pm}(z)\exp{\left[i\,\theta(x, t, z)\,\sigma_3\right]}$,
where
\begin{align}
E_{\pm}\left(z\right)=\begin{bmatrix}
1&\frac{iq_{\pm}}{z} \vspace{0.05in}\\\frac{iq_{\pm}}{z}&1
\end{bmatrix},\quad
\theta\left(x, t; z\right)=\frac{1}{2}\left(z+\frac{q_0^2}{z}\right)\left\{x+\left[\left(z-\frac{q_0^2}{z}\right)^2-2q_0^2\right]t\right\}, \quad z\not=\pm iq_0,
\end{align}

\begin{remark} As $z=\pm iq_0$, ${\rm det} E_{\pm}(z)\equiv 0$, whereas $z\not=\pm iq_0$, ${\rm det} E_{\pm}(z)=1+\frac{q_0^2}{z^2}=\gamma_f(z)\not=0$ in which its inverse matrix exists.
Moreover, we find that $ X_{\pm}T_{\pm}= T_{\pm} X_{\pm}=(4k^2-2q_0^2)X_{\pm}^2=(2q_0^2-4k^2)(q_0^2+k^2)I$, and
  \bee\label{xt}
  \begin{array}{l}
   X_{\pm}E_{\pm}(z)=\d i\lambda E_{\pm}(z)\sigma_3=\frac{i}{2}\left(z+\frac{q_0^2}{z}\right)E_{\pm}(z)\sigma_3,\\
   T_{\pm}E_{\pm}(z)=i\lambda (4k^2-2q_0^2) E_{\pm}(z)\sigma_3
   =\d\frac{i}{2}\left(z+\frac{q_0^2}{z}\right)\left[\left(z-\frac{q_0^2}{z}\right)^2-2q_0^2\right]E_{\pm}(z)\sigma_3,
   \end{array}
     \ene
which allows one to define the Jost solutions as simultaneous solutions of both parts of the Lax pair (\ref{lax-x}) and (\ref{lax-t}).
\end{remark}

\subsubsection{Properties of Jost solutions}

As usual, the continuous spectrum is the set of all values of $z$ satisfying $\lambda(z)\in\mathbb{R}$ \cite{Biondini2014}. Let $C_0=\{z\in\mathbb{C}: |z|=q_0\}$ be the circle of radius $q_0$ (see Fig.~\ref{f} (right)). Then, the continuous spectrum is $\left(\mathbb{R}\backslash\{0\}\right)\cup C_0$ denoted by $\Sigma$. We will seek for the simultaneous solutions $\varPhi_{\pm}(x, t; z)$ of the Lax pair (\ref{lax-x}, \ref{lax-t}), i.e., the so-called Jost solutions, such that
\begin{align}\label{Jost-asy}
\varPhi_{\pm}(x, t; z)=E_{\pm}(z)\,\mathrm{e}^{i\theta(x, t; z)\sigma_3}+o\left(1\right),\quad z\in\Sigma,\quad \mathrm{as}\,\, x\to\pm\infty.
\end{align}
The modified Jost solutions is introduced though dividing by the asymptotic exponential oscillations
\begin{align}\label{bianjie}
\mu_{\pm}(x, t; z)=\varPhi_{\pm}(x, t; z)\,\mathrm{e}^{-i\theta(x, t; z)\sigma_3},
\end{align}
such that
\begin{align}\label{mJost-asy}
\lim_{x\to\pm\infty}\mu_{\pm}(x, t; z)=E_{\pm}(z).
\end{align}
Then the Jost integral equation can be obtained from Eqs.~(\ref{lax-x}) by the constant variation approach
\begin{align}\label{Jost-int}
\mu_{\pm}(x, t; z)=E_{\pm}(z)+\left\{
\begin{aligned}
&\int_{\pm\infty}^xE_{\pm}(z)\,\mathrm{e}^{i\lambda(z)(x-y)\widehat\sigma_3}\left[{E_{\pm}^{-1}(z)}\Delta Q_{\pm}(y, t)\,\mu_{\pm}(y, t; z)\right]\,\mathrm{d}y, && z\ne\pm iq_0,\\[0.05in]
&\int_{\pm\infty}^x\left[I+\left(x-y\right)\left(Q_{\pm}\mp q_0\,\sigma_3\right)\right]\Delta Q_{\pm}(y, t)\,\mu_{\pm}(y, t; z)\,\mathrm{d}y, &&  z=\pm iq_0,
\end{aligned}\right.
\end{align}
where $\Delta Q_{\pm}(x, t)=Q(x, t)-Q_{\pm}$.

\begin{lemma}\label{isjie}
Given a series $\sum_{n=0}^{\infty}A_n(x)$ and a function $B(x)$ on an interval $\mathbb{D}\subset\mathbb{R}$, where $A_n(x)$ and $B(x)$ are matrix-valued functions. If $\sum_{n=0}^{\infty}\left|\left|A_n(x)\right|\right|_1$ converges uniformly on the interval $\mathbb{D}$ and $\left|\left|B(x)\right|\right|_1$, $\left|\left|B(x)\,A_n(x)\right|\right|_1\in L^1\left(\mathbb{D}\right)$, then $\left|\left|B(x)\sum_{n=0}^{\infty}A_n(x)\right|\right|_1\in L^1\left(\mathbb{D}\right)$  and
 $       \int_\mathbb{D}B(x)\sum_{n=0}^{\infty}A_n(x)\mathrm{d}x=\sum_{n=0}^\infty\int_\mathbb{D}B(x)\,A_n(x)\,\mathrm{d}x.$
\end{lemma}

\begin{proposition} \label{jiexi-m1}
Suppose $q-q_{\pm}\in L^1\left(\mathbb{R^{\pm}}\right)$, then the  Jost integral equation (\ref{Jost-int}) has unique solutions $\mu_{\pm}(x, t; z)$ defined by Eq.~(\ref{bianjie}) in $\Sigma^0:=\Sigma\backslash\{\pm iq_0\}$. Moreover, the columns $\mu_{+1}(x, t; z)$ and $\mu_{-2}(x, t; z)$ can be extended analytically to $D_{+}$ and continuously to $D_{+}\cup\Sigma^0$, and $\mu_{-1}(x, t; z)$ and $\mu_{+2}(x, t; z)$ can be extended analytically to $D_{-}$ and continuously to $D_{-}\cup\Sigma^0$, where $\mu_{\pm j}(x, t; z)\, (j=1,2)$ is the $j$-th column of $\mu_{\pm}(x, t; z)$.
\end{proposition}

\begin{corollary}\label{jiexi-1}
Suppose $q-q_{\pm}\in L^1\left(\mathbb{R^{\pm}}\right)$, then Eq.~(\ref{lax-x}) has unique solutions $\varPhi_{\pm}(x, t; z)$  defined by Eq.~(\ref{Jost-asy}) in $\Sigma^0$. Moreover, $\varPhi_{+1}(x, t; z)$ and $\varPhi_{-2}(x, t; z)$ can be extended analytically to $D_{+}$ and continuously to $D_{+}\cup\Sigma^0$, and $\varPhi_{-1}(x, t; z)$ and $\varPhi_{+2}(x, t; z)$ can be extended analytically to $D_{-}$ and continuously to $D_{-}\cup\Sigma^0$, where $\varPhi_{\pm j}(x, t; z)\, (j=1,2)$  is the $j$-th column of $\varPhi_{\pm}(x, t; z)$.
\end{corollary}

\begin{proposition} \label{jiexi-m2}
Suppose $\left(1+\left|x\right|\right)\left(q-q_{\pm}\right)\in L^1\left(\mathbb{R^{\pm}}\right)$, then the Jost integral equation (\ref{Jost-int}) has unique solutions $\mu_{\pm}(x, t; z)$ defined by Eq.~(\ref{bianjie}) in $\Sigma$. Besides, $\mu_{+1}(x, t; z)$ and $\mu_{-2}(x, t; z)$ can be extended analytically to $D_{+}$ and continuously to $D_{+}\cup\Sigma$ while $\mu_{-1}(x, t; z)$ and $\mu_{+2}(x, t; z)$ can be extended analytically to $D_{-}$ and continuously to $D_{-}\cup\Sigma$.
\end{proposition}

\begin{corollary}\label{jiexi-2}
Suppose $\left(1+\left|x\right|\right)\left(q-q_{\pm}\right)\in L^1\left(\mathbb{R^{\pm}}\right)$, then  Eq.~(\ref{lax-x}) has unique solutions $\varPhi_{\pm}(x, t; z)$ defined by (\ref{Jost-asy}) in $\Sigma$. Besides, $\varPhi_{+1}(x, t; z)$ and $\varPhi_{-2}(x, t; z)$ can be extended analytically to $D_{+}$ and continuously to $D_{+}\cup\Sigma$ while $\varPhi_{-1}(x, t; z)$ and $\varPhi_{+2}(x, t; z)$ can be extended analytically to $D_{-}$ and continuously to $D_{-}\cup\Sigma$.
\end{corollary}

\begin{lemma}[Liouville's formula]\label{Liouville}
Consider an $n$-dimensional first-order homogeneous linear ordinary differential equation,
$y'=A(x)\,y$, on an interval $\mathbb{D}\subset\mathbb{R}$, where $A(x)$ denotes a complex square matrix of order $n$.
Let $\varPhi$ be a matrix-valued solution of this equation. If the trace $\mathrm{tr}\,A(x)$ is a continuous function, then one has
$\mathrm{det}\, \varPhi(x)=\mathrm{det}\, \varPhi(x_0)\,\exp\left[\int_{x_0}^x\mathrm{tr}A(\xi)\,\mathrm{d}\xi\right],\, x, x_0\in\mathbb{D}.$
\end{lemma}

\begin{proposition}[Time evolution of the Jost solutions]
The Jost solutions $\varPhi_{\pm}(x, t; z)$ are the simultaneous solutions of both parts of the Lax pair (\ref{lax-x},  \ref{lax-t}).
\end{proposition}

\subsubsection{Scattering matrix, scattering coefficients, and reflection coefficients}

In this subsection, the scattering matrix is introduced. Since Tr$\,X(z)=\mathrm{Tr}\,T(z)=0$, Liouville's formula yields $({\rm det}\varPhi_{\pm})_x=({\rm det} \varPhi_{\pm})_t=0$. Thus ${\rm det}\,\varPhi_{\pm}={\rm det}\, E_{\pm}(z)$. It follows from Lemma \ref{Liouville} that $\varPhi_{\pm}(x, t; z)$ are two fundamental matrix solutions of Lax pair (\ref{lax-x}, \ref{lax-t}) for $\forall z\in\Sigma^0$, Thus there exists a constant matrix $S(z)$ (not depend on $x$ and $t$) such that
\begin{align}\label{Jostchuandi}
\varPhi_+(x, t; z)=\varPhi_-(x, t; z)\,S(z),
\end{align}
where $S(z)=\left(s_{ij}(z)\right)_{2\times 2}$ is referred to the scattering matrix and its entries $s_{ij}(z)$ as the scattering coefficients.
\begin{proposition}\label{jiexiS}
Suppose $q-q_{\pm}\in L^1\left(\mathbb{R^{\pm}}\right)$. Then $s_{11}(z)$ can be extended analytically to $D_+$ and continuously to $D_+\cup\Sigma^0$ while $s_{22}(z)$ can be extended analytically to $D_-$ and continuously to $D_-\cup\Sigma^0$. Moreover, both $s_{12}(z)$ and $s_{21}(z)$ are continuous in $\Sigma^0$.
\end{proposition}
\begin{proof}
It can from Eq.~(\ref{Jostchuandi}) that 
one has the following Wronskian representations for the scattering coefficients:
\begin{align}\label{S-lie}
\begin{aligned}
s_{11}(z)=\frac{{\rm Wr}(\varPhi_{+1}(x, t; z), \varPhi_{-2}(x, t; z))}{\gamma_f(z)},\quad s_{22}(z)=\frac{{\rm Wr}(\varPhi_{-1}(x, t; z), \varPhi_{+2}(x, t; z))}{\gamma_f(z)},\\[0.05in]
s_{12}(z)=\frac{{\rm Wr}(\varPhi_{+2}(x, t; z), \varPhi_{-2}(x, t; z))}{\gamma_f(z)},\quad s_{21}(z)=\frac{{\rm Wr}(\varPhi_{-1}(x, t; z), \varPhi_{+1}(x, t; z))}{\gamma_f(z)},
\end{aligned}
\end{align}
where ${\rm Wr}(\bm\cdot, \bm\cdot)$ denotes the Wronskian determinant and $\gamma_f(z):={\rm det} E_{\pm}(z)=1+\frac{q_0^2}{z^2}$. The analyticity properties of the scattering coefficients can be established from Corollary \ref{jiexi-1}.
\end{proof}

With the  help of  Proposition \ref{jiexi-2} and Eq.~(\ref{S-lie}), one can establish the following Corollary.
\begin{corollary}
Suppose $\left(1+\left|x\right|\right)\left(q-q_{\pm}\right)\in L^1\left(\mathbb{R^{\pm}}\right)$. Then $\lambda(z)\,s_{11}(z)$ can be extended analytically to $D_+$ and continuously to $D_+\cup\Sigma$ while $\lambda(z)\,s_{22}(z)$ can be extended analytically to $D_-$ and continuously to $D_-\cup\Sigma$. Moreover, both $\lambda(z)\,s_{12}(z)$ and $\lambda(z)\,s_{21}(z)$ are continuous in $\Sigma$.
\end{corollary}
Note that one cannot exclude the possible presence of zeros for $s_{11}(z)$ and $s_{22}(z)$ along $\Sigma^0$. To solve the Riemann-Hilbert problem in the inverse process, we restrict our consideration to potentials without spectral singularities \cite{Zhou1989}, i.e., $s_{11}(z)\ne 0,\, s_{22}(z)\ne 0$ for $z\in\Sigma$. The so-called reflection coefficients $\rho(z)$ and $\tilde\rho(z)$ are defined as
\begin{align}\label{fanshe}
\rho(z)=\frac{s_{21}(z)}{s_{11}(z)}
=\frac{{\rm Wr}(\varPhi_{-1}(x, t; z), \varPhi_{+1}(x, t; z))}{{\rm Wr}(\varPhi_{+1}(x, t; z), \varPhi_{-2}(x, t; z))}, \,\,
\tilde\rho(z)=\frac{s_{12}(z)}{s_{22}(z)}=\frac{{\rm Wr}(\varPhi_{+2}(x, t; z), \varPhi_{-2}(x, t; z))}{{\rm Wr}(\varPhi_{-1}(x, t; z), \varPhi_{+2}(x, t; z))},\,\, z\in\Sigma^0.
\end{align}

\subsubsection{Symmetries}

In this subsection, we study the symmetries of the Jost solutions $\varPhi_{\pm}(x, t; z)$ and scattering matrix $S(z)$. To this aim, we first pose the reduction conditions of the Lax pair on the complex $z$-plane. There are two involutions in $k$-plane: $(k, \lambda)\to (k^*,\, \lambda^*)$ and
$(k, \lambda)\to (k,\, -\lambda)$, which generate the corresponding two involutions in $z$-plane:  $z\to z^*$ and
$z\to -q_0^2/z$.

\begin{proposition}[Reduction conditions]\label{3RC}
The $X(x, t; z)$ and $T(x, t; z)$ in the Lax pair (\ref{lax-x}, \ref{lax-t}) keep  the following three reduction conditions on $z$-plane.
\begin{itemize}
\item The first reduction condition
\begin{align}\label{first-RC}
X(x, t; z)=\sigma_2\,X(x, t; z^*)^*\,\sigma_2, \quad T(x, t; z)=\sigma_2\,T(x, t; z^*)^*\,\sigma_2.
\end{align}
\item The second reduction condition
\begin{align}\label{second-RC}
X(x, t; z)=X(x, t; -z^*)^*,\quad T(x, t; z)=T(x, t; -z^*)^*.
\end{align}
\item The third reduction condition
\begin{align}\label{third-RC}
X(x, t; z)=X\left(x, t; -\frac{q_0^2}{z}\right),\quad T(x, t; z)=T\left(x, t; -\frac{q_0^2}{z}\right).
\end{align}
\end{itemize}
\end{proposition}

\begin{proposition}
The three symmetries for the Jost solutions $\varPhi_{\pm}(x, t; z)$ in $z\in\Sigma$ are given as follows:
\begin{itemize}
\item The first symmetry
\begin{align}\label{Jduichen-1}
\varPhi_{\pm}(x, t; z)&=\sigma_2\,\varPhi_{\pm}(x, t; z^*)^*\,\sigma_2,
\quad i.e.,\quad
\varPhi_{\pm j}(x, t; z)=(-1)^{j+1}i\,\sigma_2\,\varPhi_{\pm (3-j)}(x, t; z^*)^*,\quad j=1,2.
\end{align}
\item The second symmetry
\begin{align}\label{Jduichen-2}
\varPhi_{\pm}(x, t; z)&=\varPhi_{\pm}(x, t; -z^*)^*.
\end{align}

\item The third symmetry
\begin{align}\label{Jduichen-3}
\varPhi_{\pm}(x, t; z)&=\frac{iq_{\pm}}{z}\,\varPhi_{\pm}\left(x, t; -\frac{q_0^2}{z}\right)\sigma_1,
\quad i.e.,\quad
\varPhi_{\pm j}(x, t; z)=\frac{iq_{\pm}}{z}\,\varPhi_{\pm (3-j)}\left(x, t; -\frac{q_0^2}{z}\right),\quad j=1,2.
\end{align}
\end{itemize}
\end{proposition}

\begin{proposition}\label{Sduichen}
The three symmetries for the scattering matrix $S(z)$ in $z\in\Sigma$ are given as follows:
\begin{itemize}
\item  The first symmetry
\begin{align}\label{Sduichen-1}
S(z)=\sigma_2\,S(z^*)^*\,\sigma_2,\quad i.e.,\quad
s_{11}(z)=s_{22}(z^*)^*,\quad s_{12}(z)=-s_{21}(z^*)^*.
\end{align}
\item The second symmetry
\begin{align}\label{Sduichen-2}
S(z)=S(-z^*)^*, \quad i.e.,\quad
s_{ij}(z)=s_{ij}(-z^*)^*,\quad i, j=1, 2.
\end{align}
\item The third symmetry
\begin{align}\label{Sduichen-3}
S(z)=\frac{q_+}{q_-}\,\sigma_1S\left(-\frac{q_0^2}{z}\right)\sigma_1,
\quad i.e.,\quad
s_{11}(z)=\frac{q_+}{q_-}\,s_{22}\left(-\frac{q_0^2}{z}\right),\quad s_{12}(z)=\frac{q_+}{q_-}\,s_{21}\left(-\frac{q_0^2}{z}\right).
\end{align}
\end{itemize}
\end{proposition}

Note that the above-obtained symmetries for the Jost solutions $\varPhi_{\pm}(x, t; z)$ and scattering matrix $S(z)$ are established in $z\in\Sigma^0$, but the symmetries of the columns $\varPhi_{\pm 1}(x, t; z)$, $\varPhi_{\pm 2}(x, t; z)$ and the scattering coefficients $s_{11}(z), s_{22}(z)$ can hold in their extended regions.

\begin{corollary}
Three symmetries for the reflection coefficients $\rho(z)$ and $\tilde\rho(z)$ in $z\in\Sigma$ are listed below:
\begin{itemize}
\item The first symmetry
\begin{align}\label{rho-1}
\rho(z)=-\tilde\rho(z^*)^*.
\end{align}
\item The second symmetry
\begin{align}\label{rho-2}
\rho(z)=\rho(-z^*)^*,\quad \tilde\rho(z)=\tilde\rho(-z^*)^*.
\end{align}
\item The third symmetry
\begin{align}
\rho(z)=\tilde\rho\left(-\frac{q_0^2}{z}\right).
\end{align}
\end{itemize}
\end{corollary}

\subsubsection{Discrete spectrum and residue conditions with simple poles}
The discrete spectrum of the scattering problem is the set of all values $z\in \mathbb{C}\backslash\Sigma$ such that the scattering problem admits eigenfunctions in $L^2(\mathbb{R})$. As is shown in \cite{Biondini2014}, these are precisely the values of $z$ in $D_+$ for which $s_{11}(z)=0$ and those values in $D_-$ for which $s_{22}(z)=0$. In this subsection, we suppose that $s_{11}(z)$ has $N_1$ and $N_2$ simple zeros, respectively, in $D_+\cap\left\{z\in\mathbb{C}:\mathrm{Re} \,z>0, \mathrm{Im}\, z>0\right\}$ denoted by $z_n$, $n=1, 2, \cdots, N_1$ and in $i\left(q_0, +\infty\right)$ denoted by $iw_n$, $n=1, 2, \cdots, N_2$. It follows from the symmetries of the scattering matrix in Proposition \ref{Sduichen} that
\begin{align}
s_{11}(-z_n^*)=s_{22}\!\!\left(\!\!-\frac{q_0^2}{z_n}\right)\!\!=s_{22}\!\!\left(\!\frac{q_0^2}{z_n^*}\!\right)=s_{22}(z_n^*)
=s_{22}(-z_n)=s_{11}\!\!\left(\!\!-\frac{q_0^2}{z_n^*}\right)=s_{11}\!\!\left(\!\frac{q_0^2}{z_n}\right)=0,\,\, n=1,2,\cdots, N_1,
\end{align}
\begin{align}
s_{22}(-iw_n)=s_{11}\!\left(\frac{iq_0^2}{w_n}\right)=s_{22}\!\left(\!-\frac{iq_0^2}{w_n}\right)=0,\quad n=1,2,\cdots, N_2.
\end{align}
Thus, the discrete spectrum is the set
\begin{align}\label{lisanpu}
Z=\left\{z_n, z_n^*, -z_n^*, -z_n, -\frac{q_0^2}{z_n}, \frac{q_0^2}{z_n^*}, -\frac{q_0^2}{z_n^*}, \frac{q_0^2}{z_n}\right\}_{n=1}^{N_1}\bigcup\left\{iw_n, -iw_n, \frac{iq_0^2}{w_n}, -\frac{iq_0^2}{w_n}\right\}_{n=1}^{N_2},
\end{align}
whose distribution is shown in Fig. \ref{f} (right).
\begin{lemma}\label{duojidian}
Suppose $z_0$ is an isolated pole of order $m$ for $F(z)$, then $F(z)$ has the following Laurent expansion
$F(z)=\sum_{j=-m}^\infty F_j\left(z-z_0\right)^j$ with $
F_j=\lim_{z\to z_0}\frac{\left[\left(z-z_0\right)^mF(z)\right]^{\left(j+m\right)}}{\left(j+m\right)!}.$
\end{lemma}
The Lemma can be verified easily by the knowledge of complex variable \cite{Ablowitz2003}. One can use the lemma to derive the residue condition with simple poles and double poles.

As $z_0\in Z\cap D_+$, since $s_{11}(z_0)=0$, it follows from Eq.~(\ref{S-lie}) that there exists a constant $b[z_0]$ such that
\begin{align} \no
\varPhi_{+1}(x, t; z_0)=b[z_0]\,\varPhi_{-2}(x, t; z_0).
\end{align}
The residue condition is given from Lemma \ref{duojidian} by:
\begin{align}\no
\mathop\mathrm{Res}\limits_{z=z_0}\left[\frac{\varPhi_{+1}(x, t; z)}{s_{11}(z)}\right]=\frac{\varPhi_{+1}(x, t; z_0)}{s_{11}'(z_0)}=\frac{b[z_0]}{s_{11}'(z_0)}\,\varPhi_{-2}(x, t; z_0).
\end{align}
As $z_0\in Z\cap D_-$, similarly one has
\begin{align}\no
\varPhi_{+2}(x, t; z_0)=b[z_0]\,\varPhi_{-1}(x, t; z_0)
\end{align}
and
\begin{align}\no
\mathop\mathrm{Res}\limits_{z=z_0}\left[\frac{\varPhi_{+2}(x, t; z)}{s_{22}(z)}\right]=\frac{\varPhi_{+2}(x, t; z_0)}{s_{22}'(z_0)}=\frac{b[z_0]}{s_{22}'(z_0)}\,\varPhi_{-1}(x, t; z_0).
\end{align}
For  convenience, we denote the constants $b[z_0]$ and $A[z_0]$, respectively, by
\begin{align}\label{bAdingyi}
b[z_0]=\left\{
\begin{aligned}
\frac{\varPhi_{+1}(x, t; z_0)}{\varPhi_{-2}(x, t; z_0)}, \quad z_0\in Z\cap D_+,\\[0.05in]
\frac{\varPhi_{+2}(x, t; z_0)}{\varPhi_{-1}(x, t; z_0)}, \quad z_0\in Z\cap D_-,
\end{aligned}\right.\qquad
A[z_0]=\left\{
\begin{aligned}
\frac{b[z_0]}{s_{11}'(z_0)},\quad z_0\in Z\cap D_+,\\[0.05in]
\frac{b[z_0]}{s_{22}'(z_0)},\quad z_0\in Z\cap D_-,
\end{aligned}\right.
\end{align}
and the more compact form of residue conditions can be represented as
\begin{align}\label{Jie-liushu}
\begin{aligned}
\mathop\mathrm{Res}\limits_{z=z_0}\left[\frac{\varPhi_{+1}(x, t; z)}{s_{11}(z)}\right]&=A[z_0]\,\varPhi_{-2}(x, t; z_0),\quad z_0\in Z\cap D_+,\\[0.05in]
\mathop\mathrm{Res}\limits_{z=z_0}\left[\frac{\varPhi_{+2}(x, t; z)}{s_{22}(z)}\right]&=A[z_0]\,\varPhi_{-1}(x, t; z_0),\quad z_0\in Z\cap D_-,
\end{aligned}
\end{align}
where $\frac{\bm\cdot}{\bm\cdot}$ in the expression of $b[z_0]$ denotes the proportional coefficient.

For $z\in D_+$, there exists a constant $\delta >0$ such that
$\left\{\tilde z: \left| \tilde z-z\right|<\delta\right\}\subset D_+,\,  \left\{\tilde z^*: \left| \tilde z^*-z^*\right|<\delta\right\}\subset D_-$.
The analyticities for $s_{11}(\tilde z)$ in $\left\{\tilde z: \left| \tilde z-z\right|<\delta\right\}\subset D_+$
and for $s_{22}(\tilde z)$ in $\left\{\tilde z^*: \left| \tilde z^*-z^*\right|<\delta\right\}\subset D_-$ yields the Taylor expansions
\begin{align}\label{s22e}
s_{11}(\tilde z)=\sum_{m=0}^\infty \frac{s_{11}^{\left(m\right)}(z)}{m!}\left(\tilde z-z\right)^m,\,\,{\rm for}\,\, \tilde z\subset D_+, \quad
s_{22}(\tilde z^*)=\sum_{m=0}^\infty \frac{s_{22}^{\left(m\right)}(z^*)}{m!}\left(\tilde z^*-z^*\right)^m, \,\,{\rm for}\,\, \tilde z\subset D_-.
\end{align}

Next, we imply the relations for the residue between different discrete spectral points in $Z$. To this end, we first give a Lemma.

\begin{lemma}\label{dSduichen-1e}
Two relations between $s_{11}^{\left(m\right)}(z)$ and $s_{22}^{\left(m\right)}(z)$ are given below:
\begin{itemize}
\item The first relation
\begin{align}\label{dSduichen-1e}
s_{11}^{\left(m\right)}(z)=s_{22}^{\left(m\right)}(z^*)^*, \quad m\in\mathbb{N}, \quad z\in D_{+}.
\end{align}
\item The second relation
\begin{align}\label{dSduichen-2e}
\begin{gathered}
s_{11}^{\left(m\right)}(z)=\left(-1\right)^ms_{11}^{\left(m\right)}(-z^*)^*, \quad m\in\mathbb{N},\quad z\in D_{+}, \v \\
s_{22}^{\left(m\right)}(z)=\left(-1\right)^ms_{22}^{\left(m\right)}(-z^*)^*, \quad m\in\mathbb{N},\quad z\in D_{-}.
\end{gathered}
\end{align}
\end{itemize}
\end{lemma}

\begin{proposition}
As $z_0\in Z$, there exist three relations for the residue conditions:
\begin{itemize}
\item The first relation: $
A[z_0]=-A[z_0^*]^*.$
\item The second relation:
$A[z_0]=-A[-z_0^*]^*.$
\item The third relation:
$A[z_0]=\frac{z_0^2}{q_0^2}\,A\left[-\frac{q_0^2}{z_0}\right].$
\end{itemize}
\end{proposition}

\begin{corollary}
The relations between coefficients of residue conditions in $Z$ are given.
\begin{align}\no
A[z_n]=-A[-z_n^*]^*=-A[z_n^*]^*=A[-z_n]=\frac{z_n^2}{q_0^2}\,A\left[-\frac{q_0^2}{z_n}\right]=-\frac{z_n^2}{q_0^2}\,A\left[-\frac{q_0^2}{z_n^*}\right]^*=-\frac{z_n^2}{q_0^2}\,A\left[\frac{q_0^2}{z_n^*}\right]^*=\frac{z_n^2}{q_0^2}\,A\left[\frac{q_0^2}{z_n}\right],
\end{align}
\begin{align}\no
A[iw_n]=-A[-iw_n]^*=-\frac{w_n^2}{q_0^2}\,A\left[\frac{iq_0^2}{w_n}\right]=\frac{w_n^2}{q_0^2}\,A\left[-\frac{iq_0^2}{w_n}\right]^*, \quad \mathrm{Re}\, A[iw_n]=0.
\end{align}
\end{corollary}

\subsubsection{Asymptotic behaviors}

In order to propose and solve the Riemann-Hilbert problem properly, one needs to determine the asymptotic behaviors of the modified Jost solutions and  scattering data both as $z\to\infty$ and as $z\to 0$. Consider the Neumann series used in Proposition \ref{jiexi-m1},
\begin{align}
\mu_{\pm}(x, t; z)=\sum_{n=0}^\infty \mu_{\pm}^{[n]}(x, t; z)
\end{align}
with
\begin{align} \no
\mu_{\pm}^{[0]}(x, t; z)=E_{\pm}(z),\quad
\mu_{\pm}^{[n+1]}(x, t; z)=\int_{\pm\infty}^xE_{\pm}(z)\,\mathrm{e}^{i\lambda(z)(x-y)\widehat\sigma_3}\left[E_{\pm}^{-1}(z)\Delta Q_{\pm}(y, t)\,\mu_{\pm}^{[n]}(y, t; z)\right]\,\mathrm{d}y.
\end{align}
From Refs.~\cite{Bleistein1986, Biondini2014}, one can derive
\begin{align}
\begin{aligned}
&\mu_{\pm}^{[n+1], d}(x, t; z)=\frac{1}{1+\left(q_0/z\right)^2}\Bigg[\int_{\pm\infty}^x\left(\Delta Q_{\pm}(y, t)\,\mu_{\pm}^{[n], o}(y, t; z)-\frac{i\sigma_3Q_{\pm}(y, t)}{z}\,\Delta Q_{\pm}(y, t)\,\mu_{\pm}^{[n], d}(y, t; z)\right)\,\mathrm{d}y\\
&\qquad +\frac{i\sigma_3\,Q_{\pm}}{z}\int_{\pm\infty}^x\mathrm{e}^{i\lambda(x-y)\widehat\sigma_3}\left(\Delta Q_{\pm}(y, t)\,\mu_{\pm}^{[n], d}(y, t; z)-\frac{i\sigma_3\,Q_{\pm}(y, t)}{z}\,\Delta Q_{\pm}(y, t)\,\mu_{\pm}^{[n], o}(y, t; z)\right)\,\mathrm{d}y\Bigg]\\[0.05in]
&=\left\{
\begin{aligned}
&O\left(\mu_{\pm}^{[n], o}(x, t; z)\right)+O\left(\frac{\mu_{\pm}^{[n], d}(x, t; z)}{z}\right)+O\left(\frac{\mu_{\pm}^{[n], d}(x, t; z)}{z^2}\right)+O\left(\frac{\mu_{\pm}^{[n], o}(x, t; z)}{z^3}\right),&& z\to\infty,  \\[0.05in]
&O\left(z^2\mu_{\pm}^{[n], o}(x, t; z)\right)+O\left(z\mu_{\pm}^{[n], d}(x, t; z)\right)+O\left(z^2\mu_{\pm}^{[n], d}(x, t; z)\right)+O\left(z\mu_{\pm}^{[n], o}(x, t; z)\right),&& z\to 0,
\end{aligned}\right.
\end{aligned}
\end{align}
\begin{align}
\begin{aligned}
&\mu_{\pm}^{[n+1], o}(x, t; z)=\frac{1}{1+\left(q_0/z\right)^2}\Bigg[\frac{i\sigma_3\,Q_{\pm}}{z}\int_{\pm\infty}^x\left(\Delta Q_{\pm}(y, t)\,\mu_{\pm}^{[n], o}(y, t; z)-\frac{i\sigma_3\,Q_{\pm}(y, t)}{z}\,\Delta Q_{\pm}(y, t)\,\mu_{\pm}^{[n], d}(y, t; z)\right)\,\mathrm{d}y\\[0.05in]
&\qquad+\int_{\pm\infty}^x\mathrm{e}^{i\lambda(x-y)\widehat\sigma_3}\left(\Delta Q_{\pm}(y, t)\,\mu_{\pm}^{[n], d}(y, t; z)-\frac{i\sigma_3\,Q_{\pm}(y, t)}{z}\Delta Q_{\pm}(y, t)\,\mu_{\pm}^{[n], o}(y, t; z)\right)\,\mathrm{d}y\Bigg]\\[0.05in]
&=\left\{
\begin{aligned}
&O\left(\frac{\mu_{\pm}^{[n], o}(x, t; z)}{z}\right)+O\left(\frac{\mu_{\pm}^{[n], d}(x, t; z)}{z^2}\right)+O\left(\frac{\mu_{\pm}^{[n], d}(x, t; z)}{z}\right)+O\left(\frac{\mu_{\pm}^{[n], o}(x, t; z)}{z^2}\right),&& z\to\infty, \\[0.05in]
&O\left(z\mu_{\pm}^{[n], o}(x, t; z)\right)+O\left(\mu_{\pm}^{[n], d}(x, t; z)\right)+O\left(z^3\mu_{\pm}^{[n], d}(x, t; z)\right)+O\left(z^2\mu_{\pm}^{[n], o}(x, t; z)\right),&&z\to 0,
\end{aligned}\right.
\end{aligned}
\end{align}
where $\mu_{\pm}^{[n], d}$ and $\mu_{\pm}^{[n], o}$ stand for the diagonal and off-diagonal parts of $\mu_{\pm}^{[n]}$, respectively. Recall that the individual columns of $\mu_{\pm}(x, t; z)$ are analytic in different regions of the complex $z$-plane ($\mu_{+1}(x, t; z),\, \mu_{-2}(x, t; z)$ in $D_+$, and $\mu_{+2}(x, t; z), \mu_{-1}(x, t; z)$ in $D_-$). Note that although the asymptotic relations in this subsection are represented in terms of $\mu_{\pm}(x, t; z)$ rather than column-wise, however, they are to be understood as taken in the appropriate region for each column. By the induction with
\begin{align}
\begin{aligned}
\mu_{\pm}^{[0], d}(x, t; z)&=O\left(1\right),  \quad\mu_{\pm}^{[0], o}(x, t; z)=O\left(\frac{1}{z}\right), \quad z\to\infty, \\[0.05in]
\mu_{\pm}^{[0], d}(x, t; z)&=O\left(1\right),  \quad\mu_{\pm}^{[0], o}(x, t; z)=O\left(\frac{1}{z}\right), \quad z\to 0,
\end{aligned}
\end{align}
one can obtain for $\forall m\in\mathbb{N}$
\begin{align}
\begin{alignedat}{5}
\mu_{\pm}^{[2m], d}&=O\left(\frac{1}{z^m}\right), \,\,\,&\mu_{\pm}^{[2m], o}&=O\left(\frac{1}{z^{m+1}}\right), \,\,\,&\mu_{\pm}^{[2m+1], d}&=O\left(\frac{1}{z^{m+1}}\right),  \,\,\,&\mu_{\pm}^{[2m+1], o}&=O\left(\frac{1}{z^{m+1}}\right), &\,\,\,&z\to\infty, \\[0.05in]
\mu_{\pm}^{[2m], d}&=O\left(z^m\right), &\mu_{\pm}^{[2m], o}&=O\left(z^{m-1}\right), &\mu_{\pm}^{[2m+1], d}&=O\left(z^{m}\right), &\mu_{\pm}^{[2m+1], o}&=O\left(z^{m}\right), &&z\to 0.
\end{alignedat}
\end{align}
\begin{proposition}\label{Jzjianjin}
The asymptotics for the modified Jost solutions are found as
\bee\label{Jzjianjin-1}
\mu_{\pm}(x, t; z)=\left\{\begin{array}{ll} I+O\left(\dfrac{1}{z}\right),&\quad z\to\infty,\v \\
                         \dfrac{i}{z}\,\sigma_3\,Q_{\pm}+O\left(1\right),&\quad z\to 0.
                         \end{array}\right.
\ene
\end{proposition}

\begin{corollary}\label{Szjianjin}
The asymptotic behaviors for the scattering matrix are given by
\begin{alignat}{2}
S(z)&=I+O\left(\frac{1}{z}\right),  \quad &&z\to\infty, \label{ien}\\
S(z)&=\frac{q_+}{q_-}\,I+O\left(z\right),\quad &&z\to 0. \label{ien-1}
\end{alignat}
\end{corollary}

\begin{proof}
From Proposition \ref{Jzjianjin} and Eq.~(\ref{S-lie}), one can yield that
\begin{align}\no
\begin{aligned}
s_{11}(z)&=\frac{Wr\left(\varPhi_{+1}(x, t; z), \varPhi_{-2}(x, t; z)\right)}{\gamma(z)}=\frac{Wr\left(\mu_{+1}(x, t; z), \mu_{-2}(x, t; z)\right)}{1+q_0^2/z^2}\\[0.1in]
&=\left\{
\begin{aligned}
&\frac{\mathrm{det}
\begin{pmatrix}
1+O(1/z)&O(1/z)\\[0.05in]
O(1/z)&1+O(1/z)
\end{pmatrix}}
{1+O(1/z^2)}
=1+O\left(\frac{1}{z}\right),\quad z\to\infty,\\[0.1in]
&\frac{\mathrm{det}
\begin{pmatrix}
O(1)& (i/z)q_-\\[0.05in]
(i/z)q_+&O(1)
\end{pmatrix}}
{q_0^2+O(z^2)}\,z^2=\frac{q_+}{q_-}+O(z),\quad z\to 0.
\end{aligned}\right.
\end{aligned}
\end{align}
The asymptotic behaviors for $s_{22}(z)$, $s_{12}(z)$ and $s_{21}(z)$ can also be obtained similarly. Here we omit them.
\end{proof}

\subsection{Inverse scattering problem with NZBCs}

\subsubsection{Generalized matrix Riemann-Hilbert problem}

To formulate the inverse problem as a generalized matrix Riemann-Hilbert problem (RHP), one needs to  pose a relation along $\Sigma$ based on rearranging the terms in Eq.~(\ref{Jostchuandi}). From their asymptotic behaviors and the Plemelj's formulae, the solutions for the Riemann-Hilbert problem can be proposed. Explicitly, we present the following proposition.
\begin{proposition}
Define the sectionally meromorphic matrices
\begin{align}\label{RHP-M}
M(x, t; z)=\left\{
\begin{aligned}
M^+(x, t; z)=\left(\frac{\mu_{+1}(x, t; z)}{s_{11}(z)},\, \mu_{-2}(x, t; z)\right),\quad z\in D_+, \\[0.05in]
 M^-(x, t; z)=\left(\mu_{-1}(x, t; z),\, \frac{\mu_{+2}(x, t; z)}{s_{22}(z)}\right), \quad z\in D_-.
\end{aligned}\right.
\end{align}
Then the multiplicative matrix Riemann-Hilbert problem is proposed as follows:
\begin{itemize}
\item Analyticity: $M(x, t; z)$ is analytic in $\left(D_+\cup D_-\right)\backslash Z$ and has simple poles in $Z$.
\item Jump condition:
\begin{align}\label{RHP-Jump}
M^-(x, t; z)=M^+(x, t; z)\left(I-J(x, t; z)\right), \quad z\in\Sigma,
\end{align}
where the partial jump matrix is defined by
\begin{align} \no
J(x, t; z)=\mathrm{e}^{i\theta(x, t; z)\widehat\sigma_3}
\begin{bmatrix}
0&-\tilde\rho(z)\\[0.05in]
\rho(z)&\rho(z)\,\tilde\rho(z)
\end{bmatrix}.
\end{align}
\item Asymptotic behavior:
\bee\label{RHP-Asy}
M^{\pm}(x, t; z)=\left\{\begin{array}{ll} 
 I+O\left(\dfrac{1}{z}\right), & z\to\infty, \v \\
 \dfrac{i}{z}\,\sigma_3\,Q_-+O\left(1\right), & z\to 0.
 \end{array}\right.
\ene
\end{itemize}
\end{proposition}

To solve  the above-mentioned Riemann-Hilbert problem conveniently, we introduce $\widehat\eta_n=-\frac{q_0^2}{\eta_n}$ with
\begin{align}
\eta_n=\left\{
\begin{aligned}
&z_n, & n&=1, 2, \cdots, N_1,\\[0.05in]
-{}&z^*_{n-N_1}, &n&=N_1+1, N_1+2, \cdots, 2N_1,\\[0.05in]
-{}&\frac{q_0^2}{z^*_{n-2N_1}}, &n&=2N_1+1, 2N_1+2, \cdots, 3N_1,\\[0.05in]
&\frac{q_0^2}{z_{n-3N_1}},&n&=3N_1+1,3N_1+2, \cdots, 4N_1,\\[0.05in]
&iw_{n-4N_1},& n&=4N_1+1, 4N_1+2, \cdots, 4N_1+N_2,\\[0.05in]
-{}&\frac{iq_0^2}{w_{n-4N_1-N2}},&n&=4N_1+N_2+1, 4N_1+N_2+2, \cdots, 4N_1+2N_2.
\end{aligned}\right.
\end{align}

\begin{theorem}
The solution of the Riemann-Hilbert problem given by Eqs.~(\ref{RHP-M}, \ref{RHP-Jump}, \ref{RHP-Asy}) is given by
\begin{align}\label{RHP-jie}
\begin{aligned}
M(x, t; z)=I+\frac{i}{z}\,\sigma_3\,Q_-&+\sum_{n=1}^{4N_1+2N_2}\left[\frac{\mathop\mathrm{Res}\limits_{z=\eta_n}M^+(x, t; z)}{z-\eta_n}+\frac{\mathop\mathrm{Res}\limits_{z=\widehat\eta_n}M^-(x, t; z)}{z-\widehat\eta_n}\right]\\[0.05in]
&+\frac{1}{2\pi i}\int_\Sigma\frac{M^+(x, t; \zeta)\,J(x, t; \zeta)}{\zeta-z}\,\mathrm{d}\zeta,\quad z\in\mathbb{C}\backslash\Sigma,
\end{aligned}
\end{align}
where $\int_\Sigma$ denotes the integral along the oriented contour shown in Fig. \ref{f}(right).
\end{theorem}

\begin{proof}
By subtracting out the asymptotic behaviors and  pole contributions, one can regularize the jump condition (\ref{RHP-Jump}) as
\begin{align}\label{jumpbianxing}
\begin{aligned}
&M^-(x, t; z)-I-\frac{i}{z}\,\sigma_3\,Q_--\sum_{n=1}^{4N_1+2N_2}\left[\frac{\mathop\mathrm{Res}\limits_{z=\eta_n}M^+(z)}{z-\eta_n}+\frac{\mathop\mathrm{Res}\limits_{z=\widehat\eta_n}M^-(z)}{z-\widehat\eta_n}\right]\\[0.05in]
&\quad =M^+(x, t; z)
-I-\frac{i}{z}\,\sigma_3\,Q_--\sum_{n=1}^{4N_1+2N_2}\left[\frac{\mathop\mathrm{Res}\limits_{z=\eta_n}M^+(z)}{z-\eta_n}+\frac{\mathop\mathrm{Res}\limits_{z=\widehat\eta_n}M^-(z)}{z-\widehat\eta_n}\right]
-M^+(x, t; z)\,J(x, t; z).
\end{aligned}
\end{align}

The left-hand side of Eq.~(\ref{jumpbianxing}) is analytic in $D_-$, and the right-hand side  Eq.~(\ref{jumpbianxing}) except for the last term $M^+(x, t; z)\,J(x, t; z)$, is analytic in $D_+$. Both of their asymptotics are $O\left(\frac{1}{z}\right)$ as $z\to\infty$ and $O(1)$ as $z\to 0$. From Corollary \ref{Szjianjin}, $J(x, t; z)$ is $O\left(\frac{1}{z}\right)$  as $z\to\infty$, and $O(z)$ as $z\to 0$. Hence, the Cauchy projectors $P_{\pm}$ over $\Sigma$ can be well-defined as follows:
\begin{align}
P_{\pm}\left[f\right](z)=\frac{1}{2\pi i}\int_\Sigma\frac{f(\zeta)}{\zeta-(z\pm i0)}\,\mathrm{d}\zeta,
\end{align}
where the notation $z\pm i0$ represents the limit taken from the left/right of $z$.  Applying the Cauchy projectors to Eq.~(\ref{jumpbianxing}) and using the Plemelj's formulae, one derives the solution (\ref{RHP-jie}) of the Riemann-Hilbert problem (\ref{RHP-M}, \ref{RHP-Jump}, \ref{RHP-Asy}).
\end{proof}

\subsubsection{Reconstruction formula for the potential}

To present a closed algebraic-integral system of equations for the solution of the Riemann-Hilbert problem (\ref{RHP-jie}), one needs to determine the expressions for residue conditions in Eq.~(\ref{RHP-jie}). Eqs.~(\ref{bianjie}) and (\ref{Jie-liushu}) imply that
\begin{align}
\begin{aligned}
\mathop\mathrm{Res}\limits_{z=\eta_n}\left[\frac{\mu_{+1}(x, t; z)}{s_{11}(z)}\right]&=A[\eta_n]\,\mu_{-2}(x, t; \eta_n)\,\mathrm{e}^{-2i\theta(x, t; \eta_n)}, \v\\
\mathop\mathrm{Res}\limits_{z=\widehat\eta_n}\left[\frac{\mu_{+2}(x, t; z)}{s_{22}(z)}\right]&=A[\widehat\eta_n]\,\mu_{-1}(x, t; \widehat\eta_n)\,\mathrm{e}^{2i\theta(x, t; \widehat\eta_n)}.
\end{aligned}
\end{align}
From the definition of the $M(x, t; z)$ in Eq.~(\ref{RHP-M}), one yields that only the first column has a simple pole at $z=\eta_n$ and only the second column has a simple pole at $z=\widehat\eta_n$. Then the residue parts in Eq.~(\ref{RHP-jie}) are calculated as
\begin{align}\label{liushuhe}
\frac{\mathop\mathrm{Res}\limits_{z=\eta_n}M^+(x, t; z)}{z-\eta_n}+\frac{\mathop\mathrm{Res}\limits_{z=\widehat\eta_n}M^-(x, t; z)}{z-\widehat\eta_n}=\left[C_n(z)\,\mu_{-2}(x, t; \eta_n), \widehat C_n(z)\,\mu_{-1}(x, t; \widehat\eta_n)\right],
\end{align}
where
\begin{align} \no
C_n(z)=\frac{A[\eta_n]\,\mathrm{e}^{-2i\theta(x, t; \eta_n)}}{z-\eta_n}, \quad \widehat C_n(z)=\frac{A[\widehat\eta_n]\,\mathrm{e}^{2i\theta(x, t; \widehat\eta_n)}}{z-\widehat\eta_n}.
\end{align}
Combining Eq.~(\ref{RHP-jie}) and (\ref{liushuhe}), one yields
\begin{align}\label{kuaile}
\mu_{-2}(x, t; z)=
\begin{bmatrix}
\frac{iq_-}{z}\\[0.05in]1
\end{bmatrix}
+\sum_{n=1}^{4N_1+2N_2}\widehat C_n(z)\,\mu_{-1}(x, t; \widehat\eta_n)+\frac{1}{2\pi i}\int_\Sigma\frac{\left(M^+J\right)_2(x, t; \zeta)}{\zeta-z}\,\mathrm{d}\zeta.
\end{align}
From Eq.~(\ref{Jduichen-3}), one can obtain
\begin{align}\label{kuaile-1}
\mu_{-2}(x, t; z)=\frac{iq_-}{z}\,\mu_{-1}\left(x, t; -\frac{q_0^2}{z}\right).
\end{align}
Letting $z=\eta_k, k=1, 2, \cdots, 4N_1+2N_2$ in Eqs.~(\ref{kuaile}, \ref{kuaile-1}), then one has
\begin{align}\label{yaoqiujiele}
\begin{bmatrix}
\frac{iq_-}{\eta_k}\\[0.05in]1
\end{bmatrix}
+\sum_{n=1}^{4N_1+2N_2}\left(\widehat C_n(\eta_k)-\frac{iq_-}{\eta_k}\,\delta_{k, n}\right)\mu_{-1}(x, t; \widehat\eta_n)+\frac{1}{2\pi i}\int_\Sigma\frac{\left(M^+J\right)_2(x, t; \zeta)}{\zeta-\eta_k}\,\mathrm{d}\zeta=0,
\end{align}
where $\delta_{k, n}$ is the Kronecker delta function. These equations for $k=1, 2, \cdots, 4N_1+2N_2$ comprise a system of $4N_1+2N_2$ equations with $4N_1+2N_2$ unknowns $\mu_{-1}(x, t; \widehat\eta_n), n=1, 2, \cdots, 4N_1+2N_2$, which together with Eqs.~(\ref{RHP-jie}) and (\ref{kuaile-1}), give a closed system of equations for $M(x, t; z)$ in terms of the scattering data.

Next, we construct the potential from the solution of the Riemann-Hilbert problem.

\begin{theorem}
The potential with single poles in the focusing mKdV equation with NZBCs is given by
\begin{align}\label{recons}
q(x, t)=q_--i\sum_{n=1}^{4N_1+2N_2}A[\widehat\eta_n]\,\mathrm{e}^{2i\theta(x, t; \widehat\eta_n)}\,\mu_{-11}(x, t; \widehat\eta_n)+\frac{1}{2\pi}\int_\Sigma\left(M^+J\right)_{12}(x, t; \zeta)\,\mathrm{d}\zeta.
\end{align}
\end{theorem}
\begin{proof}
Since $M(x, t; z)\,\mathrm{e}^{i\theta(x, t; z)\sigma_3}$ solves  Eq.~(\ref{lax-x}),  it follows that
\begin{align}\label{haha}
M_x(x, t; z)+M(x, t; z)\left(\frac{i\sigma_3}{2}\,z+\frac{iq_0^2\,\sigma_3}{2\,z}\right)=\left(\frac{i\sigma_3}{2}\,z-\frac{iq_0^2\,\sigma_3}{2\,z}+Q\right)M(x, t; z).
\end{align}
From Eqs.~(\ref{RHP-jie}) and (\ref{liushuhe}), one obtains the asymptotic behavior of $M(x, t; z)$ as
\begin{align}
M(x, t; z)=I+\frac{1}{z}\,M^{(1)}(x, t; z)+O\left(\frac{1}{z^2}\right), \quad z\to\infty,
\end{align}
where
\begin{align}
\begin{aligned}
M^{(1)}(x, t; z)&=i\,\sigma_3\,Q_--\frac{1}{2\pi i}\int_\Sigma M^+(x, t; \zeta)\,J(x, t; \zeta)\,\mathrm{d}\zeta  \\[0.05in]
&\qquad +\sum_{n=1}^{4N_1+2N_2}\left[A[\eta_n]\,\mathrm{e}^{-2i\theta(x, t; \eta_n)}\mu_{-2}(x, t; \eta_n), \,A[\widehat\eta_n]\,\mathrm{e}^{2i\theta(x, t; \widehat\eta_n)}\mu_{-1}(x, t; \widehat\eta_n)\right].
\end{aligned}
\end{align}
The proof follows by comparing with the coefficient of $z^0$.
\end{proof}

\subsubsection{Trace formulae and theta condition}
The so-called trace formula is that the scattering coefficients $s_{11}(z)$ and $s_{22}(z)$ are formulated in terms of the discrete spectrum $Z$ and reflection coefficients $\rho(z)$ and $\tilde\rho(z)$. Recall that $s_{11}(z)$ is analytic in $D_+$ and $s_{22}(z)$ is analytic in $D_-$. The discrete spectral points $\eta_n$ is the simple zeros of $s_{11}(z)$, while $\widehat\eta_n$ is the simple zeros of $s_{22}(z)$.

Let
\begin{align}
\begin{aligned}
\beta^+(z)=s_{11}(z)\prod_{n=1}^{4N_1+2N_2}\frac{z-\widehat\eta_n}{z-\eta_n}, \qquad
\beta^-(z)=s_{22}(z)\prod_{n=1}^{4N_1+2N_2}\frac{z-\eta_n}{z-\widehat\eta_n}.
\end{aligned}
\end{align}
Then one can yield that $\beta^+(z)$ and $\beta^-(z)$ are analytic and have no zeros in $D_+$ and $D_-$, respectively. Moreover, Eq.~(\ref{ien}) implies the asymptotic behavior: $\beta^{\pm}(z)\to 1$ as $z\to\infty$. Taking the determinants of both sides of Eq.~(\ref{Jostchuandi}) yields $\mathrm{det}\,S(z)=s_{11}(z)\,s_{22}(z)-s_{12}(z)\,s_{21}(z)=1$, with which one has $\beta^+(z)\,\beta^-(z)=\frac{1}{1-\rho(z)\,\tilde\rho(z)}$. And then taking its logarithms becomes $-\log\beta^+(z)-\log\beta^-(z)=\log\left[1-\rho(z)\,\tilde\rho(z)\right].$
With the aid of the Cauchy projectors and Plemelj's formulae, one has
\begin{align}
\log\beta^{\pm}(z)=\mp\frac{1}{2\pi i}\int_\Sigma\frac{\log\left[1-\rho(\zeta)\,\tilde\rho(\zeta)\right]}{\zeta-z}\,\mathrm{d}\zeta, \quad z\in D^{\pm}.
\end{align}
Hence, the trace formulae are given in the following:
\begin{align}\label{trace-1}
s_{11}(z)&=\exp\left(-\frac{1}{2\pi i}\int_\Sigma\frac{\log\left[1-\rho(\zeta)\,\tilde\rho(\zeta)\right]}{\zeta-z}\,\mathrm{d}\zeta\right)\prod_{n=1}^{4N_1+2N_2}\frac{z-\eta_n}{z-\widehat\eta_n},
\\[0.05in]
s_{22}(z)&=\exp\left(\frac{1}{2\pi i}\int_\Sigma\frac{\log\left[1-\rho(\zeta)\,\tilde\rho(\zeta)\right]}{\zeta-z}\,\mathrm{d}\zeta\right)\prod_{n=1}^{4N_1+2N_2}\frac{z-\widehat\eta_n}{z-\eta_n}.
\end{align}

In the following, we use the obtained trace formulae to derive the asymptotic phase difference of the boundary values $q_+$ and $q_-$ (also called `theta condition' in Ref.~\cite{Faddeev1987}). To this end, let $z\to 0$ in Eq.~(\ref{trace-1}). The left-hand side of Eq.~(\ref{ien-1}) yields $s_{11}(z)\to\frac{q_+}{q_-}$.
Note that
\begin{align}
\prod_{n=1}^{4N_1+2N_2}\frac{z-\eta_n}{z-\widehat\eta_n}=
\prod_{n=1}^{N_1}\frac{\left(z-z_n\right)\left(z+z_n^*\right)\left(z+\dfrac{q_0^2}{z_n^*}\right)\left(z-\dfrac{q_0^2}{z_n}\right)}
{\left(z-z_n^*\right)\left(z+z_n\right)\left(z+\dfrac{q_0^2}{z_n^*}\right)\left(z-\dfrac{q_0^2}{z_n^*}\right)}\,
\prod_{m=1}^{N_2}\frac{\left(z-iw_m\right)\left(z+\dfrac{iq_0^2}{w_m}\right)}{\left(z+iw_m\right)\left(z-\dfrac{iq_0^2}{w_m}\right)}.
\end{align}
One can infer that
\begin{align}
\prod_{n=1}^{4N_1+2N_2}\frac{z-\eta_n}{z-\widehat\eta_n}\to 1 \quad \mathrm{as}\quad z\to 0.
\end{align}
Furthermore, one has
\begin{align} \label{qq}
\frac{q_+}{q_-}=\exp\left(\frac{i}{2\pi}\int_\Sigma\frac{\log\left[1-\rho(\zeta)\,\tilde\rho(\zeta)\right]}{\zeta}\,\mathrm{d}\zeta\right).
\end{align}
Thus, the theta condition for Eq.~(\ref{qq}) is given in the following:
\begin{align}
\mathrm{arg}\,\frac{q_+}{q_-}=\frac{1}{2\pi}\int_\Sigma\frac{\log\left[1-\rho(\zeta)\,\tilde\rho(\zeta)\right]}{\zeta}\,\mathrm{d}\zeta.
\end{align}
Further, let $f(\zeta)=\log\left[1-\rho(\zeta)\,\tilde\rho(\zeta)\right]$,
\begin{gather*}
L_1=\int_{-\infty}^{-q_0}\frac{f(\zeta)}{\zeta}\,\mathrm{d}\zeta, \quad L_2=\int_{q_0}^{+\infty}\frac{f(\zeta)}{\zeta}\,\mathrm{d}\zeta, \quad L_3=\int_{0}^{-q_0}\frac{f(\zeta)}{\zeta}\,\mathrm{d}\zeta, \quad L_4=\int_{q_0}^{0}\frac{f(\zeta)}{\zeta}\,\mathrm{d}\zeta,\\
K_1=i\int_{\frac{\pi}{2}}^0f(q_0\mathrm{e}^{i\varphi})\,\mathrm{d}\varphi,\quad K_2=i\int_{\pi}^\frac{\pi}{2}f(q_0\mathrm{e}^{i\varphi})\,\mathrm{d}\varphi,\quad K_3=i\int_{-\pi}^{-\frac{\pi}{2}}f(q_0\mathrm{e}^{i\varphi})\,\mathrm{d}\varphi,\quad K_4=i\int_{-\frac{\pi}{2}}^0f(q_0\mathrm{e}^{i\varphi})\,\mathrm{d}\varphi.
\end{gather*}
The symmetries in Eqs. (\ref{rho-1}, \ref{rho-2}) yield  $f(\zeta)=f(-\zeta)$, which generates $L_1=-L_2,\, L_3=-L_4,\, K_1=-K_3,\, K_2=-K_4$. Then one has
\begin{gather}
\mathrm{arg}\,\frac{q_+}{q_-}=\frac{1}{2\pi}\int_\Sigma\frac{\log\left[1-\rho(\zeta)\,\tilde\rho(\zeta)\right]}{\zeta}\,\mathrm{d}\zeta
=\frac{1}{2\pi}\sum_{j=1}^4\left(L_j+K_j\right)=0,
\end{gather}
that is $q_+=q_-$, which means that the boundary conditions  at infinity are same.

\subsubsection{Reflectionless potential}

In this subsection, we will explicitly exhibit the IST  with the aid of the Riemann-Hilbert problem. We here consider a special  kind of solutions, where the reflection coefficients $\rho(z)$ and $\tilde\rho(z)$ vanish identically. In this case, there is no jump (i.e., $J=0$) from $M^+(x, t; z)$ to $M^-(x, t; z)$ along  the continuous spectrum, and the inverse problem can be solved explicitly by using an algebraic system.

The case $\rho(z)=\tilde\rho(z)=0$ implies $J(x, t; z)=0$. It follows from Eq.~(\ref{yaoqiujiele}) that
\begin{align}\label{sle}
\sum_{n=1}^{4N_1+2N_2}\left(\widehat C_n(\eta_k)-\frac{iq_-}{\eta_k}\,\delta_{k, n}\right)\mu_{-11}(x, t; \widehat\eta_n)=-\frac{iq_-}{\eta_k}, \quad k=1, 2, \cdots, 4N_1+2N_2.
\end{align}
Let $G=\left(g_{kn}\right)_{(4N_1+2N_2)\times(4N_1+2N_2)}$, $\gamma=\left(\gamma_n\right)_{(4N_1+2N_2)\times 1}$, $\beta=\left(\beta_k\right)_{(4N_1+2N_2)\times 1}$ with
$g_{kn}=\widehat C_n(\eta_k)-\frac{iq_-}{\eta_k}\,\delta_{k, n}, \, \gamma_n=\mu_{-11}(x, t; \widehat\eta_n), \, \beta_k=-\frac{iq_-}{\eta_k}.$
Then one can obtain $\gamma=G^{-1}\beta$ by solving the system of linear equations (\ref{sle}).

Let $\alpha=\left(\alpha_n\right)_{(4N_1+2N_2)\times 1}$, where $\alpha_n=A[\widehat\eta_n]\,\mathrm{e}^{2i\theta(x, t; \widehat\eta_n)}$.
From the reconstruction formula in Eq.~(\ref{recons}), we have the following theorem:

\begin{theorem}
The reflectionless potential (i.e., the single-pole solution of the focusing mKdV equation with NZBCs (\ref{mKdV})) can be derived via the determinants
\begin{align}\label{solu1}
q(x, t)=q_-+\frac{\mathrm{det} \begin{bmatrix} G&\beta \vspace{0.05in}\\
\alpha^T&0 \end{bmatrix}} {\mathrm{det}\,G}\,i.
\end{align}
\end{theorem}

The reflectionless potential contains free parameters $N_1, N_2, q_-, z_n, A[z_n], w_m, A[iw_m]$, $(n=1, 2, \cdots N_1; m=1, 2, \cdots N_2)$. The solution (\ref{solu1}) possesses the distinct wave structures for some parameters:

\begin{figure}[!t]
\centering
\includegraphics[scale=0.42]{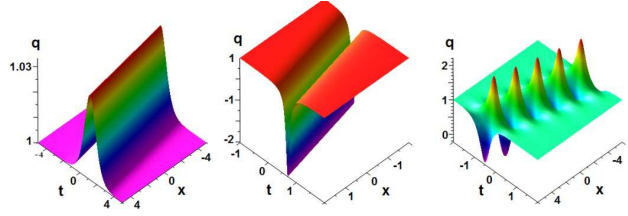}
\caption{Soliton and breather solutions of the fousing mKdV equation with NZBCs $q_{\pm}=1$. Left: bright soliton with parameters $N_1=0, N_2=1, w_1=\frac{6}{5}, A[iw_1]=i$. Middle: dark soliton with parameters $N_1=0, N_2=1,  w_1=\frac{6}{5}, A[iw_1]=-i$. Right: breather solution with parameters $N_1=1, N_2=0, z_1=2+\frac{1}{2}\,i, A[z_1]=i$.}
\label{f1}
\end{figure}

\begin{figure}[!t]
\centering
\includegraphics[scale=0.42]{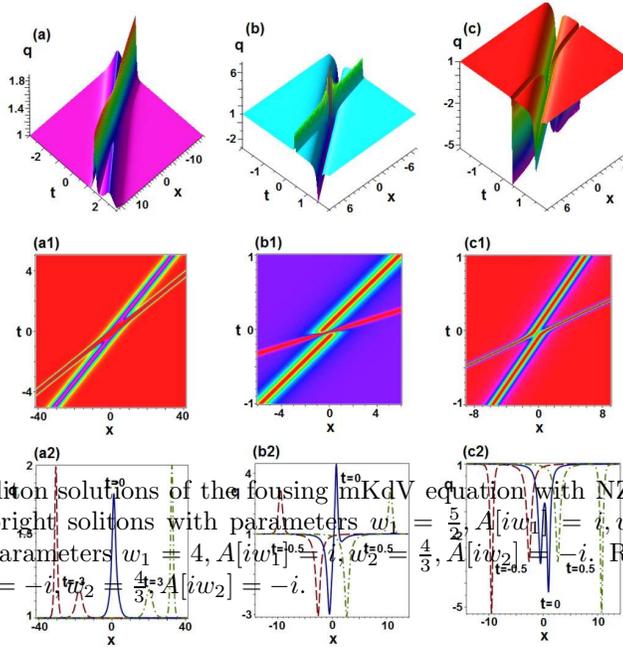}
\vspace{-1.1in}
\caption{Simple-pole $2$-soliton solutions of the fousing mKdV equation with NZBCs $q_+=q_-=1$ and $N_1=0,\, N_2=2$. Left: bright-bright solitons with parameters $w_1=\frac{5}{2}, A[iw_1]=i, w_2=\frac{3}{2}, A[iw_2]=i$. Middle: bright-dark solitons  with parameters $w_1=4, A[iw_1]=i, w_2=\frac{4}{3}, A[iw_2]=-i$. Right: dark-dark solitons  with parameters $w_1=4, A[iw_1]=-i, w_2=\frac{4}{3}, A[iw_2]=-i$.}
\label{f2-1}
\end{figure}

\begin{itemize}

 \item {} As $N_1=0,\, N_2\not=0$, it exhibits an $N_2$-soliton solution;

  \item {} As $N_1\not=0,\, N_2=0$, it stands for an $N_1$-breather solution;

  \item{} As $N_1N_2\not=0$, it is a mixed $N_1$-breather-$N_2$-soliton solution.

  \end{itemize}

 Since the mKdV equation admits a scaling symmetry, that is, if $q(x,t)$ is a solution of Eq.~(\ref{mKdV}), so is $\alpha q(\alpha x, \alpha^3t)$ with $\alpha\in \mathbb{R}\backslash\{0\}$.

 Nowadays, we explicitly show some special wave structures (e.g., single solitons and interactions of two solitons) for the solution (\ref{solu1}) as follows:
 \begin{itemize}

 \item {} As $N_1=0,\, N_2=1$, the one-soliton solution of Eq.~(\ref{mKdV}) reads as
 \begin{gather}
q(x,t)=q_-+\frac{4ca^3q_0^2(a+2)^3\mathrm{e}^\varphi}
{c^2\left(a^2+2a+2\right)^2\mathrm{e}^{2\varphi}+4a(a+2)(a+1)^2q_-\left[2c\mathrm{e}^\varphi+a(a+2)q_-\right]},
\end{gather}
where
\begin{gather*}
a=\left(\mathrm{Im}\,w_1-1\right)q_0, \,\, c=\mathrm{Im}\,A[w_1],\,\, \varphi=\frac{a(a+2)q_0}{\left(a+1\right)^3}\left[\left(a+1\right)^2x-\left(a^4+4a^3+10a^2+12a+6\right)q_0^2t\right].
\end{gather*}
As $c>0$, it is a bright soliton with NZBCs (see Fig. \ref{f1}(left)). As $c<0$, it is a dark soliton (see Fig. \ref{f1}(middle)).

 \item {} As $N_1=1,\, N_2=0$, we show the dynamical structure of the breather solution (see Fig.~\ref{f1}(right));

 \item {} As $N_1=0,\, N_2=2$, Fig.~\ref{f2-1} displays the elastic collisions of two bright solitons (left), a dark soliton and a bright soliton (middle), and two dark solitons (right) of the focusing mKdV equation with NZBCs  $q_{\pm}=1$ for distinct parameters, respectively. It follows from Fig.~\ref{f2-1}(a2) that during the interaction of two bright solitons, the bright soliton with a higher amplitude is located behind another bright soliton with a lower amplitude before their collision, and after their collision, the bright soliton with a higher amplitude moves in front of another bright soliton with a lower amplitude. Similarly, there are the same situations for another two cases (see
     Figs.~\ref{f2-1}(b2, c2)). The centerlines of two bright solitons are both located two lines, respectively, before and after interactions (see Fig.~\ref{f2-1}(a1)). For the interaction of a dark soliton and a bright soliton (see Fig.~\ref{f2-1}(b1)), the two centerlines of the bright soliton is a line, but two centerlines of the dark soliton are two parallel lines before and after interactions. Similarly, for the interaction of two dark solitons (see Fig.~\ref{f2-1}(c1)), the two centerlines of the dark soliton with narrow wave width is a line, but
     two centerlines of another dark soliton with narrow wave width are two parallel lines before and after interactions.

 \item{} As $N_1=2,\, N_2=0$, Fig.~\ref{f2-2}(left) exhibits the interaction of two breather solutions  of the focusing mKdV equation with NZBCs  $q_{\pm}=1$. The two centerlines of the breather solution with low amplitude is a line, but two centerlines of another breather solution with high amplitude are two parallel lines before and after interactions.

 \item{} As $N_1=N_2=1$, Figs.~\ref{f2-2}(middle, right) display the interaction of a breather and a bright soliton (middle), and the interaction of a breather and a dark soliton (right) of the focusing mKdV equation with NZBCs $q_{\pm}=1$, respectively. For the interaction of a breather and a bright soliton (see Fig.~\ref{f2-2}(middle)), their centerlines are both located on two lines, respectively, before and after interactions, except for the case nearby the interaction point. The interaction of a breather and a dark soliton ((see Fig.~\ref{f2-2} (right)) also admits the similar result.

  \end{itemize}

\begin{figure}[!t]
\centering
\includegraphics[scale=0.42]{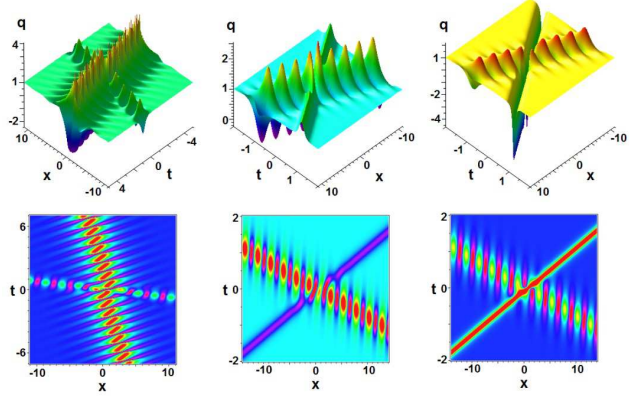}
\caption{ Simple-pole solutions of the fousing mKdV equation with NZBCs $q_+=q_-=1$. Left: $2$-breather solution with parameters $N_1=2, N_2=0, z_1=2+i/2, A[z_1]=i, z_2=1+i, A[z_2]=i$. Middle: breather-bright-soliton solutions with parameters $N_{1,2}=1, z_1=2+i/2, A[z_1]=i, w_1=2, A[iw_1]=i$. Right: breather-dark-soliton solutions with parameters $N_{1,2}=1, z_1=2+i/2, A[z_1]=i, w_1=2, A[iw_1]=-i$.}
\label{f2-2}
\end{figure}

In fact, the reflection coefficients $\rho(z),\, \tilde{\rho}(z)$ may possess the multiple poles except for the above-mentioned case of simple pole. In the following, we will consider that the reflection coefficients admit the case of double poles with pairs of conjugate complex discrete spectra and pure imaginary discrete spectra.

\section{The focusing mKdV equation with NZBCs: double poles}

\subsection{Direct scattering problem with NZBCs}

Most of the direct scattering is changeless by the presence of double poles in contrast to simple poles apart from the treatment of the discrete spectrum. Recall that the discrete spectrum is the set
\begin{align}
Z=\left\{z_n, z_n^*, -z_n^*, -z_n, -\frac{q_0^2}{z_n}, \frac{q_0^2}{z_n^*}, -\frac{q_0^2}{z_n^*}, \frac{q_0^2}{z_n}\right\}_{n=1}^{N_1}\bigcup\left\{iw_n, -iw_n, \frac{iq_0^2}{w_n}, -\frac{iq_0^2}{w_n}\right\}_{n=1}^{N_2}.
\end{align}
In this part, we suppose that the discrete spectral points are double zeros of the scattering coefficients $s_{11}(z)$ and $s_{22}(z)$, that is, we have $s_{11}(z_0)=s_{11}'(z_0)=0$, $s_{11}''(z_0)\ne 0$ for $\forall z_0\in Z\cap D_+$, and $s_{22}(z_0)=s_{22}'(z_0)=0$, $s_{22}''(z_0)\ne 0$  for $\forall z_0\in Z\cap D_-$.

Let
\begin{align}\no
b[z_0]=\left\{
\begin{aligned}
\frac{\varPhi_{+1}(x, t; z_0)}{\varPhi_{-2}(x, t; z_0)}, \quad z_0\in Z\cap D_+,\\[0.05in]
\frac{\varPhi_{+2}(x, t; z_0)}{\varPhi_{-1}(x, t; z_0)}, \quad z_0\in Z\cap D_-,
\end{aligned}\right.\qquad
d[z_0]=\left\{
\begin{aligned}
\frac{\varPhi_{+1}'(x, t; z_0)-b[z_0]\,\varPhi_{-2}'(x, t; z_0)}{\varPhi_{-2}(x, t; z_0)},\quad z_0\in Z\cap D_+,\\[0.05in]
\frac{\varPhi_{+2}'(x, t; z_0)-b[z_0]\,\varPhi_{-1}'(x, t; z_0)}{\varPhi_{-1}(x, t; z_0)},\quad z_0\in Z\cap D_-,
\end{aligned}\right.
\end{align}
\begin{align}\no
A[z_0]=\left\{
\begin{aligned}
\frac{2\,b[z_0]}{s_{11}''(z_0)},\quad z_0\in Z\cap D_+,\\[0.05in]
\frac{2\,b[z_0]}{s_{22}''(z_0)},\quad z_0\in Z\cap D_-,
\end{aligned}\right. \qquad\qquad\qquad\qquad
B[z_0]=\left\{
\begin{aligned}
\frac{d[z_0]}{b[z_0]}-\frac{s_{11}'''(z_0)}{3\,s_{11}''(z_0)},\quad z_0\in Z\cap D_+,\\[0.05in]
\frac{d[z_0]}{b[z_0]}-\frac{s_{22}'''(z_0)}{3\,s_{22}''(z_0)},\quad z_0\in Z\cap D_-.
\end{aligned}\right.
\end{align}
From Lemma \ref{duojidian}, one has
\begin{align}\label{erjieliu}
\begin{aligned}
\mathop\mathrm{P_{-2}}\limits_{z=z_0}\left[\frac{\varPhi_{+1}(x, t; z)}{s_{11}(z)}\right]&=A[z_0]\,\varPhi_{-2}(x, t; z_0), \quad z_0\in Z\cap, D_+,\\[0.05in]
\mathop\mathrm{P_{-2}}\limits_{z=z_0}\left[\frac{\varPhi_{+2}(x, t; z)}{s_{22}(z)}\right]&=A[z_0]\,\varPhi_{-1}(x, t; z_0), \quad z_0\in Z\cap, D_-,\\[0.05in]
\mathop\mathrm{Res}\limits_{z=z_0}\left[\frac{\varPhi_{+1}(x, t; z)}{s_{11}(z)}\right]&=A[z_0]\left[\varPhi_{-2}'(x, t; z_0)+B[z_0]\, \varPhi_{-2}(x, t; z_0)\right], \quad z_0\in Z\cap D_+,\\[0.05in]
\mathop\mathrm{Res}\limits_{z=z_0}\left[\frac{\varPhi_{+2}(x, t; z)}{s_{22}(z)}\right]&=A[z_0]\left[\varPhi_{-1}'(x, t; z_0)+B[z_0]\, \varPhi_{-1}(x, t; z_0)\right], \quad z_0\in Z\cap D_-,
\end{aligned}
\end{align}
where $\mathop\mathrm{P_{-2}}\limits_{z=z_0}\left[\bm\cdot\right]$ denotes the coefficient of $\frac{1}{\left(z-z_0\right)^2}$ in the Laurent expansion of $\bm\cdot$ at $z=z_0$.

\begin{proposition}
For $\forall z_0\in Z$, three symmetry relations for $A[z_0]$ and $B[z_0]$ are given by
\begin{itemize}
\item The first symmetry relation:
$A[z_0]=-A[z_0^*]^*, \quad B[z_0]=B[z_0^*]^*.$
\item The second symmetry relation:
$A[z_0]=A[-z_0^*]^*, \quad B[z_0]=-B[-z_0^*]^*.$
\item The third symmetry relation:
$A[z_0]=\frac{z_0^4}{q_0^4}\,A\left[-\frac{q_0^2}{z_0}\right],\quad B[z_0]=\frac{q_0^2}{z_0^2}\,B\left[-\frac{q_0^2}{z_0}\right]+\frac{2}{z_0}.$
\end{itemize}
\end{proposition}
\begin{corollary}
For $n=1, 2, \cdots, N_1$ and $m=1, 2, \cdots, N_2$, one has
\begin{gather}\no
A[z_n]=-A[z_n^*]^*=-A[-z_n]=A[-z_n^*]^*=\frac{z_n^4}{q_0^4}\,A\left[-\frac{q_0^2}{z_n}\right]=\frac{z_n^4}{q_0^4}\,A\left[\frac{q_0^2}{z_n^*}\right]^*
=-\frac{z_n^4}{q_0^4}\,A\left[-\frac{q_0^2}{z_n^*}\right]^*=-\frac{z_n^4}{q_0^4}\,A\left[\frac{q_0^2}{z_n}\right], \\[0.05in]
\no A[iw_m]=-A[-iw_m]^*=\frac{w_m^4}{q_0^4}\,A\left[\frac{iq_0^2}{w_m}\right]=-\frac{w_m^4}{q_0^4}\,A\left[-\frac{iq_0^4}{w_m}\right]^*, \quad \mathrm{Im}\,A[iw_m]=0,  \\[0.05in]
\begin{aligned}\no
B[z_n]=B[z_n^*]^*=&-B[-z_n^*]=-B[-z_n]=\frac{q_0^2}{z_n^2}\,B\left[-\frac{q_0^2}{z_n}\right]+\frac{2}{z_n}
        =\frac{q_0^2}{z_n^2}\,B\left[-\frac{q_0^2}{z_n^*}\right]^*+\frac{2}{z_n}\\[0.05in]
={}&-\frac{q_0^2}{z_n^2}\,B\left[\frac{q_0^2}{z_n^*}\right]^*+\frac{2}{z_n}=-\frac{q_0^2}{z_n^2}\,B\left[\frac{q_0^2}{z_n}\right]+\frac{2}{z_n},
\end{aligned}\\[0.05in]
\no B[iw_m]=B[-iw_m]^*=-\frac{q_0^2}{w_m^2}\,B\left[\frac{iq_0^2}{w_m}\right]-\frac{2i}{w_m}=-\frac{q_0^2}{w_m^2}\,B\left[-\frac{iq_0^2}{w_m}\right]^*-\frac{2i}{w_m},\quad \mathrm{Re}\,B[iw_m]=0.
\end{gather}
\end{corollary}

\subsection{Inverse problem  with NZBCs and double poles}

\subsubsection{Formulation of the RHP}

In the case of double poles, the Riemann-Hilbert problem (\ref{RHP-M}, \ref{RHP-Jump}, \ref{RHP-Asy}) still holds. To regularize the Riemann-Hilbert problem, one has to subtract out the asymptotic values as $z\to\infty$ and $z\to 0$ and the singularity contributions. Then the jump condition (\ref{RHP-Jump}) becomes
\begin{align}
\begin{aligned}
&M^--I-\frac{i}{z}\,\sigma_3\,Q_--\sum_{n=1}^{4N_1+2N_2}\left[\frac{\mathop\mathrm{P_{-2}}\limits_{z=\eta_n}M^+}{\left(z-\eta_n\right)^2}+\frac{\mathop\mathrm{Res}\limits_{z=\eta_n}M^+}{z-\eta_n}+\frac{\mathop\mathrm{P_{-2}}\limits_{z=\widehat\eta_n}M^-}{\left(z-\widehat\eta_n\right)^2}+\frac{\mathop\mathrm{Res}\limits_{z=\widehat\eta_n}M^-}{z-\widehat\eta_n}\right]\\[0.05in]
&\qquad =M^+-I-\frac{i}{z}\,\sigma_3\,Q_--\sum_{n=1}^{4N_1+2N_2}\left[\frac{\mathop\mathrm{P_{-2}}\limits_{z=\eta_n}M^+}{\left(z-\eta_n\right)^2}+\frac{\mathop\mathrm{Res}\limits_{z=\eta_n}M^+}{z-\eta_n}+\frac{\mathop\mathrm{P_{-2}}\limits_{z=\widehat\eta_n}M^-}{\left(z-\widehat\eta_n\right)^2}+\frac{\mathop\mathrm{Res}\limits_{z=\widehat\eta_n}M^-}{z-\widehat\eta_n}\right]
-M^+\,J.
\end{aligned}
\end{align}

Applying the Cauchy projectors and the Plemelj's formulae, we give the integral representation for the solution of the Riemann-Hilbert problem in the following theorem:

\begin{theorem}
The solution for the Riemann-Hilbert problem with double poles is given as
\begin{align}\label{RHP-jie-1}
\begin{aligned}
M(x, t; z)=I+\frac{i}{z}\,\sigma_3\,Q_-&+\sum_{n=1}^{4N_1+2N_2}\left[\frac{\mathop\mathrm{P_{-2}}\limits_{z=\eta_n}M^+}{\left(z-\eta_n\right)^2}+\frac{\mathop\mathrm{Res}\limits_{z=\eta_n}M^+}{z-\eta_n}+\frac{\mathop\mathrm{P_{-2}}\limits_{z=\widehat\eta_n}M^-}{\left(z-\widehat\eta_n\right)^2}+\frac{\mathop\mathrm{Res}\limits_{z=\widehat\eta_n}M^-}{z-\widehat\eta_n}\right] \\[0.05in]
&+\frac{1}{2\pi i}\int_\Sigma\frac{M^+(x, t; \zeta)\,J(x, t; \zeta)}{\zeta-z}\,\mathrm{d}\zeta,\quad z\in\mathbb{C}\backslash\Sigma.
\end{aligned}
\end{align}
\end{theorem}

\subsubsection{Closed system for the solution of RHP}

To further express the solution for the Riemann-Hilbert problem, one needs to evaluate the parts of $\mathop\mathrm{P_{-2}}(\cdot)$ and Res $(\cdot)$ appearing in Eq.~(\ref{RHP-jie-1}). By Eq.~(\ref{erjieliu}), one can obtain
\begin{align}
\begin{aligned}
\mathop\mathrm{P_{-2}}\limits_{z=\eta_n}\left[\frac{\mu_{+1}(x, t; z)}{s_{11}(z)}\right]&=A[\eta_n]\,\mathrm{e}^{-2i\theta(x, t; \eta_n)}\mu_{-2}(x, t; \eta_n),\\[0.05in]
\mathop\mathrm{P_{-2}}\limits_{z=\widehat\eta_n}\left[\frac{\mu_{+2}(x, t; z)}{s_{22}(z)}\right]&=A[\widehat\eta_n]\,\mathrm{e}^{2i\theta(x, t; \widehat\eta_n)}\mu_{-1}(x, t; \widehat\eta_n),\\[0.05in]
\mathop\mathrm{Res}\limits_{z=\eta_n}\left[\frac{\mu_{+1}(x, t; z)}{s_{11}(z)}\right]&=A[\eta_n]\,\mathrm{e}^{-2i\theta(x, t; \eta_n)}\left\{\mu_{-2}'(x, t; \eta_n)+\Big[B[\eta_n]-2\,i\,\theta'(x, t; \eta_n)\Big]\mu_{-2}(x, t; \eta_n)\right\}, \\[0.05in]
\mathop\mathrm{Res}\limits_{z=\widehat\eta_n}\left[\frac{\mu_{+2}(x, t; z)}{s_{22}(z)}\right]&=A[\widehat\eta_n]\,\mathrm{e}^{2i\theta(x, t; \widehat\eta_n)}\left\{\mu_{-1}'(x, t; \widehat\eta_n)+\Big[B[\widehat\eta_n]+2\,i\,\theta'(x, t; \widehat\eta_n)\Big]\mu_{-1}(x, t; \widehat\eta_n)\right\}.
\end{aligned}
\end{align}
Therefore, the parts of $\mathop\mathrm{P_{-2}}(\cdot)$ and Res $(\cdot)$ appearing in Eq.~(\ref{RHP-jie-1}) can be derived as
\begin{gather}\label{iznal}
\begin{aligned}
&\frac{\mathop\mathrm{P_{-2}}\limits_{z=\eta_n}M^+}{\left(z-\eta_n\right)^2}+\frac{\mathop\mathrm{Res}\limits_{z=\eta_n}M^+}{z-\eta_n}+\frac{\mathop\mathrm{P_{-2}}\limits_{z=\widehat\eta_n}M^-}{\left(z-\widehat\eta_n\right)^2}+\frac{\mathop\mathrm{Res}\limits_{z=\widehat\eta_n}M^-}{z-\widehat\eta_n}\\
&\quad=\left(C_n(z)\left[\mu_{-2}'(\eta_n)+\left(D_n+\frac{1}{z-\eta_n}\right)\mu_{-2}(\eta_n)\right],
\,\widehat C_n(z)\left[\mu_{-1}'(\widehat\eta_n)+\left(\widehat D_n+\frac{1}{z-\widehat\eta_n}\right)\mu_{-1}(\widehat\eta_n)\right]\right),
\end{aligned}
\end{gather}
where
\begin{gather}
C_n(z)=\frac{A[\eta_n]}{z-\eta_n}\,\mathrm{e}^{-2i\theta(\eta_n)}, \,\, D_n=B[\eta_n]-2\,i\,\theta'(\eta_n), \,\, \widehat C_n(z)=\frac{A[\widehat\eta_n]}{z-\widehat\eta_n}\,\mathrm{e}^{2i\theta(\widehat\eta_n)}, \,\, \widehat D_n=B[\widehat\eta_n]+2\,i\,\theta'(\widehat\eta_n).
\end{gather}
The next task is to evaluate $\mu_{-2}'(\eta_n), \mu_{-2}(\eta_n), \mu_{-1}'(\eta_n)$, and  $\mu_{-1}(\eta_n)$. It follows from Eq.~(\ref{RHP-jie-1}) with Eq.~(\ref{iznal}) that the second column of Eq.~(\ref{RHP-jie-1}) yields
\begin{align}\label{usb-1}
\mu_{-2}(z)=
\begin{bmatrix}
\frac{iq_-}{z}\\[0.05in]1
\end{bmatrix}
+\sum_{n=1}^{4N_1+2N_2}\widehat C_n(z)\left[\mu_{-1}'(\widehat\eta_n)+\left(\widehat D_n+\frac{1}{z-\widehat\eta_n}\right)\mu_{-1}(\widehat\eta_n)\right]+\frac{1}{2\pi i}\int_\Sigma\frac{\left(M^+J\right)_2(\zeta)}{\zeta-z}\,\mathrm{d}\zeta.
\end{align}
Taking the first-order derivative of $\mu_{-2}(z)$ with respect to $z$, it becomes
\begin{align}\label{usb-2}
\mu_{-2}'(z)=-
\begin{bmatrix}
\frac{iq_-}{z^2}\\[0.05in]0
\end{bmatrix}
-\sum_{n=1}^{4N_1+2N_2}\frac{\widehat C_n(z)}{z-\widehat\eta_n}\left[\mu_{-1}'(\widehat\eta_n)+\left(\widehat D_n+\frac{2}{z-\widehat\eta_n}\right)\mu_{-1}(\widehat\eta_n)\right]+\frac{1}{2\pi i}\int_\Sigma\frac{\left(M^+J\right)_2(\zeta)}{\left(\zeta-z\right)^2}\,\mathrm{d}\zeta.
\end{align}
Taking the first-order derivative of both sides in Eq.~(\ref{kuaile-1}) with respect to $z$, one gives
\begin{align}\label{kuaile-2}
\mu_{-2}'(z)=-\frac{iq_-}{z^2}\,\mu_{-1}\left(-\frac{q_0^2}{z}\right)+\frac{iq_-q_0^2}{z^3}\,\mu_{-1}'\left(-\frac{q_0^2}{z}\right).
\end{align}
Substituting Eqs.~(\ref{kuaile-1}) and (\ref{kuaile-2}) into Eqs.~(\ref{usb-1}) and (\ref{usb-2}), respectively, and letting $z=\eta_k, k=1, 2, \cdots, 4N_1+2N_2$, we obtain a linear system of $8N_1+4N_2$ equations with the $8N_1+4N_2$ unknowns $\mu_{-1}(\widehat\eta_n), \mu_{-1}'(\widehat\eta_n), n=1, 2, \cdots, 4N_1+2N_2$ in the form
\begin{gather}\label{xiama}
\begin{gathered}
\begin{aligned}
&\sum_{n=1}^{4N_1+2N_2}\widehat C_n(\eta_k)\,\mu_{-1}'(\widehat\eta_n)+\left[\widehat C_n(\eta_k)\left(\widehat D_n+\frac{1}{\eta_k-\widehat\eta_n}\right)-\frac{iq_-}{\eta_k}\,\delta_{k, n}\right]\mu_{-1}(\widehat\eta_n)\\[0.05in]
&\qquad\qquad =-\begin{bmatrix}
\frac{iq_-}{\eta_k}\\[0.05in]1
\end{bmatrix}
-\frac{1}{2\pi i}\int_\Sigma\frac{\left(M^+J\right)_2(\zeta)}{\zeta-\eta_k}\,\mathrm{d}\zeta,
\end{aligned}  \\[0.1in]
\begin{aligned}
&\sum_{n=1}^{4N_1+2N_2}\left(\frac{\widehat C_n(\eta_k)}{\eta_k-\widehat\eta_n}+\frac{iq_-q_0^2}{\eta_k^3}\,\delta_{k, n}\right)\mu_{-1}'(\widehat\eta_n)+\left[\frac{\widehat C_n(\eta_k)}{\eta_k-\widehat\eta_n}\left(\widehat D_n+\frac{2}{\eta_k-\widehat\eta_n}\right)-\frac{iq_-}{\eta_k^2}\,\delta_{k, n}\right]\mu_{-1}(\widehat\eta_n)\\[0.05in]
&\qquad \qquad =-\begin{bmatrix}
\frac{iq_-}{\eta_k^2}\\[0.05in]0
\end{bmatrix}
+\frac{1}{2\pi i}\int_\Sigma\frac{\left(M^+J\right)_2(\zeta)}{\left(\zeta-\eta_k\right)^2}\,\mathrm{d}\zeta.
\end{aligned}
\end{gathered}
\end{gather}
Solving the linear system of equations and using Eqs.~(\ref{kuaile-1}, \ref{iznal}, \ref{kuaile-2}) yields  a closed integral system for the solution $M(x, t; z)$ in Eq.~(\ref{RHP-jie-1}) of the Riemann-Hilbert problem.

\subsubsection{Reconstruction formula for the potential}

From Eqs.~(\ref{RHP-jie-1}, \ref{iznal}), the asymptotic behavior of $M(x, t; z)$ is obtained as follows.
\begin{gather}
M(x, t; z)=I+\frac{1}{z}\,M^{(1)}(x, t; z)+O\left(\frac{1}{z^2}\right), \quad z\to\infty,
\end{gather}
where
\begin{align}\no
\begin{aligned}
&M^{(1)}(x, t; z)=i\,\sigma_3\,Q_--\frac{1}{2\pi i}\int_\Sigma \left(M^+J\right)(\zeta)\,\mathrm{d}\zeta   \\[0.05in]
&\qquad +\sum_{n=1}^{4N_1+2N_2}\left[A[\eta_n]\,\mathrm{e}^{-2i\theta(\eta_n)}\Big(\mu_{-2}'(\eta_n)+D_n\,\mu_{-2}(\widehat\eta_n)\Big), \, A[\widehat\eta_n]\,\mathrm{e}^{2i\theta(\widehat\eta_n)}\left(\mu_{-1}'(\widehat\eta_n)+\widehat D_n\,\mu_{-1}(\widehat\eta_n)\right)\right].
\end{aligned}
\end{align}
Substituting $M(x, t; z)\mathrm{e}^{-i\theta(x, t; z)\sigma_3}$ into Eq.~(\ref{lax-x}) and comparing the coefficient, one has the following theorem:

\begin{theorem}
The reconstruction formula for the potential with double poles of the focusing mKdV equation with NZBCs is given by
\begin{align}\label{jinyew}
q(x, t)=q_--i\sum_{n=1}^{4N_1+2N_2}A[\widehat\eta_n]\,\mathrm{e}^{2i\theta(\widehat\eta_n)}\left(\mu_{-11}'(\widehat\eta_n)+\widehat D_n\,\mu_{-11}(\widehat\eta_n)\right)+\frac{1}{2\pi}\int_\Sigma\left(M^+J\right)_{12}(\zeta)\,\mathrm{d}\zeta,
\end{align}
where $\mu_{-11}'(\widehat\eta_n)$ and $\mu_{-11}(\widehat\eta_n)$, $n=1, 2, \cdots, 4N_1+2N_2$ are given by Eq.~(\ref{xiama}).
\end{theorem}

\subsubsection{Trace formula and theta condition}

Note that $\eta_n$ and $\widehat\eta_n$ are double zeros of the scattering coefficients $s_{11}(z)$ and $s_{22}(z)$, respectively. Introduce new functions
\begin{align}\label{ncuka}
\begin{aligned}
\beta^+(z)=s_{11}(z)\prod_{n=1}^{4N_1+2N_2}\left(\frac{z-\widehat\eta_n}{z-\eta_n}\right)^2, \quad
\beta^-(z)=s_{22}(z)\prod_{n=1}^{4N_1+2N_2}\left(\frac{z-\eta_n}{z-\widehat\eta_n}\right)^2.
\end{aligned}
\end{align}
One can find that $\beta^+(z)$ is analytic  and has no zeros in $D_+$, while $\beta^-(z)$ is analytic  and has no zeros in $D_-$. Their asymptotic behaviors are both $O(1)$ as $z\to\infty$. As same as the case of simple poles, by applying Cauchy projectors and the Plemelj's formulae, $\beta^{\pm}(z)$ can be solved as
\begin{align}\no
\log\beta^{\pm}(z)=\mp\frac{1}{2\pi i}\int_\Sigma\frac{\log\left[1-\rho(\zeta)\,\tilde\rho(\zeta)\right]}{\zeta-z}\,\mathrm{d}\zeta, \quad z\in D^{\pm}.
\end{align}
Using Eq.~(\ref{ncuka}), one implies the trace formulae as
\begin{align}\no
s_{11}(z)&=\exp\left(-\frac{1}{2\pi i}\int_\Sigma\frac{\log\left[1-\rho(\zeta)\,\tilde\rho(\zeta)\right]}{\zeta-z}\,\mathrm{d}\zeta\right)\prod_{n=1}^{4N_1+2N_2}
\left(\frac{z-\eta_n}{z-\widehat\eta_n}\right)^2,\\[0.05in]
\no s_{22}(z)&=\exp\left(\frac{1}{2\pi i}\int_\Sigma\frac{\log\left[1-\rho(\zeta)\,\tilde\rho(\zeta)\right]}{\zeta-z}\,\mathrm{d}\zeta\right)\prod_{n=1}^{4N_1+2N_2}\left(\frac{z-\widehat\eta_n}{z-\eta_n}\right)^2.
\end{align}
Furthermore, the theta condition can be derived as $z\to 0$. Note that
\begin{align}\no
\prod_{n=1}^{4N_1+2N_2}\left(\frac{z-\eta_n}{z-\widehat\eta_n}\right)^2=
\left[\prod_{n=1}^{N_1}\frac{\left(z-z_n\right)\left(z+z_n^*\right)\left(z+\dfrac{q_0^2}{z_n^*}\right)
\left(z-\dfrac{q_0^2}{z_n}\right)}{\left(z-z_n^*\right)\left(z+z_n\right)\left(z+\dfrac{q_0^2}{z_n^*}\right)
\left(z-\dfrac{q_0^2}{z_n^*}\right)}\,\prod_{n=1}^{N_2}\frac{\left(z-iw_n\right)\left(z+\dfrac{iq_0^2}{w_n}\right)}
{\left(z+iw_n\right)\left(z-\dfrac{iq_0^2}{w_n}\right)}\right]^2.
\end{align}
Then
\begin{align}\no
\prod_{n=1}^{4N_1+2N_2}\left(\frac{z-\eta_n}{z-\widehat\eta_n}\right)^2\to 1 \quad \mathrm{as}\quad z\to 0.
\end{align}
Thus, we obtain the asymptotic phase difference as
\begin{align}
\mathrm{arg}\,\frac{q_+}{q_-}=\frac{1}{2\pi}\int_\Sigma\frac{\log\left[1-\rho(\zeta)\,\tilde\rho(\zeta)\right]}{\zeta}\,\mathrm{d}\zeta,
\end{align}
which is same as the case of simple poles. Thus, $\mathrm{arg}\,\frac{q_+}{q_-}=0$, i.e., $q_+=q_-$.

\subsubsection{Double-pole multi-breather-soliton solutions}

In this subsection, we present the reflectionless potential of the focusing mKdV equation with NZBCs and double poles.

 Let $\rho(z)=\tilde\rho(z)=0$. Then $J(x, t; z)=0$. It follows from Eq.~(\ref{xiama}) that one can obtain a linear system of $8N_1+4N_2$  equations as
\begin{gather}\label{xiama2}
\begin{gathered}
\sum_{n=1}^{4N_1+2N_2}\widehat C_n(\eta_k)\,\mu_{-11}'(\widehat\eta_n)+\left[\widehat C_n(\eta_k)\left(\widehat D_n+\frac{1}{\eta_k-\widehat\eta_n}\right)-\frac{iq_-}{\eta_k}\,\delta_{k, n}\right]\mu_{-11}(\widehat\eta_n)
=-\frac{iq_-}{\eta_k},  \\
\sum_{n=1}^{4N_1+2N_2}\!\!\left(\frac{\widehat C_n(\eta_k)}{\eta_k-\widehat\eta_n}+\frac{iq_-q_0^2}{\eta_k^3}\,\delta_{k, n}\right)\mu_{-11}'(\widehat\eta_n)+\left[\frac{\widehat C_n(\eta_k)}{\eta_k-\widehat\eta_n}\left(\widehat D_n+\frac{2}{\eta_k-\widehat\eta_n}\right)-\frac{iq_-}{\eta_k^2}\,\delta_{k, n}\right]\mu_{-11}(\widehat\eta_n)=-\frac{iq_-}{\eta_k^2},
\end{gathered}
\end{gather}
which can be rewritten in the matrix form: $H\gamma=\beta,$ where
\bee\begin{array}{l} \no
\displaystyle
H=\begin{bmatrix}
H^{(11)}&H^{(12)}\\
H^{(21)}&H^{(22)}
\end{bmatrix},\quad
\gamma=
\begin{bmatrix}
\gamma^{(1)}\\
\gamma^{(2)}
\end{bmatrix},\quad
\beta=
\begin{bmatrix}
\beta^{(1)}\\
\beta^{(2)}
\end{bmatrix}, \v\\
\displaystyle H^{(i  j)}=\left(h^{(i  j)}_{kn}\right)_{\left(4N_1+2N_2\right)\times \left(4N_1+2N_2\right)}, \quad \gamma^{(j)}=\left(\gamma^{(j)}_n\right)_{4N_1+2N_2}, \quad \beta^{(i)}=\left(\beta^{(i)}_k\right)_{4N_1+2N_2}, \quad i, j=1, 2,     \v\\
\displaystyle h^{(11)}_{kn}=\widehat C_n(\eta_k)\left(\widehat D_n+\frac{1}{\eta_k-\widehat\eta_n}\right)-\frac{iq_-}{\eta_k}\,\delta_{k, n}, \quad h^{(12)}_{kn}=\widehat C_n(\eta_k), \quad     \v \\
\displaystyle h^{(21)}_{kn}=\frac{\widehat C_n(\eta_k)}{\eta_k-\widehat\eta_n}\left(\widehat D_n+\frac{2}{\eta_k-\widehat\eta_n}\right)-\frac{iq_-}{\eta_k^2}\,\delta_{k, n}, \quad h^{(22)}_{kn}=\frac{\widehat C_n(\eta_k)}{\eta_k-\widehat\eta_n}+\frac{iq_-q_0^2}{\eta_k^3}\,\delta_{k, n},    \v \\
\displaystyle \gamma^{(1)}_n=\mu_{-11}(\widehat\eta_n), \quad \gamma^{(2)}_n=\widehat\mu_{-11}(\widehat\eta_n), \quad \beta^{(1)}_k=-\frac{iq_-}{\eta_k}, \quad \beta^{(2)}_k=-\frac{iq_-}{\eta_k^2}.
\end{array}
\ene
Then one can obtain the solution as $\gamma=H^{-1}\beta$. We have the following theorem from Eq.~(\ref{jinyew}):

\begin{theorem}
The reflectionless potential (i.e., the solution of the focusing mKdV equation with NZBCs and double poles) is given in the determinantal form
\begin{gather}\label{solu2}
q(x, t)=q_-+
\frac{\mathrm{det}
\begin{bmatrix}
H&\beta \vspace{0.06in}\\
\alpha^T&0
\end{bmatrix}}
{\mathrm{det}\,H}\,i,
\end{gather}
where
\begin{gather}\no
\alpha=\begin{bmatrix}
\alpha^{(1)}\\
\alpha^{(2)}
\end{bmatrix},\quad
\alpha^{(j)}=\left(\alpha^{(j)}_n\right)_{4N_1+2N_2}, \, j=1,2, \quad \alpha^{(1)}_n=A[\widehat\eta_n]\,\mathrm{e}^{2i\theta(\widehat\eta_n)}\widehat D_n,  \quad \alpha^{(2)}_n=A[\widehat\eta_n]\,\mathrm{e}^{2i\theta(\widehat\eta_n)}.
\end{gather}
\end{theorem}

The reflectionless potential with double poles (\ref{solu2}) has the following free parameters: $N_1, N_2, q_-, z_n, A[z_n], B[z_n]$,  $w_m,  A[iw_m], B[iw_m], n=1, 2, \cdots, N_1, m=1, 2, \cdots, N_2$, and exhibits the distinct wave structures for different parameters:

\begin{itemize}
 \item {} As $N_1=0,\, N_2\not=0$, it displays an $N_2$-(bright, dark)-soliton solution;

  \item {} As $N_1\not=0,\, N_2=0$, it stands for an $N_1$-(breather, breather) solution;

  \item{} As $N_1\not=0,\, N_2\not=0$, it is a general $N_1$-(breather, breather)-$N_2$-(bright, dark)-soliton solution.

  \end{itemize}

\begin{figure}[!t]
\centering
\includegraphics[scale=0.42]{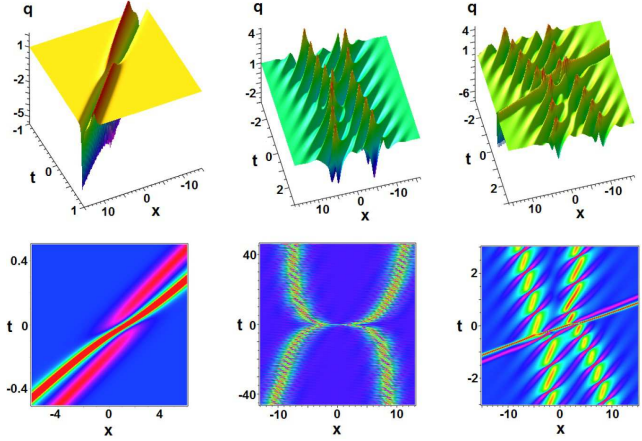}
\caption{Double-pole soliton solutions of the fousing mKdV equation with NZBCs $q_{\pm}=1$. Left: interaction of bright-dark solitons with $N_1=0, N_2=1,  w_1=3, A[iw_1]=1, B[iw_1]=i$. Middle: interaction of breather-breather solutions with $N_1=1, N_2=0,  z_1=1+1.1\,i, A[z_1]=B[z_1]=i$. Right: interaction of (breather, breather)-(bright, dark) solutions with $N_1=N_2=1, z_1=1+\,i, A[z_1]=B[z_1]=i, w_1=3, A[iw_1]=1, B[iw_1]=\frac{1}{2}\,i$.}
\label{double}
\end{figure}

  In the following, we show some special wave structures for the solution (\ref{solu2}) with double poles by choosing special parameters
  as follows:
 \begin{itemize}

 \item {} As $N_1=0,\, N_2=1$, Fig.~\ref{double} (left) displays the interaction of  a dark soliton and  a bright soliton with the same velocity
  of the focusing mKdV equation with NZBCs $q_{\pm}=1$, obtained only by pure imaginary discrete spectral points.

 \item {} As $N_0=1,\, N_2=0$, Fig.~\ref{double} (middle) exhibits the interaction of two breather solutions with the same velocity of the focusing mKdV equation with NZBCs $q_{\pm}=1$, obtained only via pairs of conjugate complex discrete spectral points.

 \item{} As $N_1=N_2=1$,  Fig.~\ref{double} (right) shows the interaction of the bright-dark soliton (see Fig.~\ref{double} (left))  and  two-breather solution (see Fig.~\ref{double} (middle)) of the focusing mKdV equation with NZBCs $q_{\pm}=1$, found by using both pure imaginary discrete spectral points and pairs of conjugate complex discrete spectral points.
  \end{itemize}

\section{The defocusing mKdV equation with NZBCs}

In this section, we focus on the study of the IST for the defocusing mKdV  equation (\ref{mKdV}) with NZBCs.

\subsection{Direct scattering  problem with NZBCs}

\subsubsection{Riemann surface and uniformization variable}

Considering the asymptotic scattering problem ($x\to \pm\infty$) of the defocusing Lax pair (\ref{lax-x}) and (\ref{lax-t}):
\begin{align}\left\{
\begin{aligned}
\varPhi_x&=X_{\pm}\varPhi, & X_{\pm}(k)&=ik\sigma_3+Q_{\pm} , \\[0.04in]
\varPhi_t&=T_{\pm}\varPhi, &T_{\pm}(k)&=\left(4k^2+2q_0^2\right)X_{\pm}(k),
\end{aligned}\right.
\end{align}
one can obtain the fundamental matrix solution as
\begin{align}
\varPhi^{bg}(x, t; k)=
\left\{
\begin{alignedat}{2}
&E_{\pm}(k)\,\mathrm{e}^{i\theta(x, t, k)\sigma_3}, &&k\ne\pm q_0,\\[0.04in]
&I+\left(x+6\,q_0^2\,t\right)X_{\pm}(k),&\quad&k=\pm q_0,
\end{alignedat}\right.
\end{align}
where
\begin{align*}
Q_{\pm}=
\begin{bmatrix}
0&q_{\pm}\\
q_{\pm}&0
\end{bmatrix},\quad
E_{\pm}(k)=
\begin{bmatrix}
1&\frac{iq_{\pm}}{k+\lambda}\\-\frac{iq_{\pm}}{k+\lambda}&1
\end{bmatrix},\quad
\theta(x, t; k)=\lambda(k)\left[x+\left(4k^2+2q_0^2\right)t\right], \quad \lambda(k)=\sqrt{k^2-q_0^2}.
\end{align*}

Since $\lambda(k)$ is doubly branched, one needs to introduce the two-sheeted Riemann surface so that $\lambda(k)$ is single-valued on this surface, where the branch points are $k=\pm q_0$ (see, e.g., Refs.\cite{zab, Faddeev1987, Prinari2006,Biondini2016a,Demontis2013}).  Let $k\mp q_0=r_{\mp}\,\mathrm{e}^{i\theta_{\mp}}$, then we have $\lambda_1(k)=\sqrt{r_-r_+}\,\mathrm{e}^{i\frac{\theta_-+\theta_+}{2}}$ on Sheet-I and $\lambda_2(k)=-\lambda_1(k)$ on Sheet-II. By restricting the arguments $-\pi\le\theta_{\mp}<\pi$, two single-valued branches are posed.  With these convention, the branch cut is determined as the segment $\left[-q_0,q_0\right]$. The two-sheeted Riemann surface is obtained by gluing the Sheet-I and Sheet-II along the branch cut. The region where Im $\lambda(k)>0$ is the UHP on the Sheet-I and the LHP on Sheet-II. The region where Im $\lambda(k)<0$ is the LHP on the Sheet-I and UHP on Sheet-II. Besides, $\lambda(k)$ is real-valued on $\left(-\infty, -q_0\right]\cup \left[q_0, +\infty\right)$. Introduce the uniformization variable $z$  defined by:
$z=k+\lambda=k+\sqrt{k^2-q_0^2}$
with the inverse mapping given by
$k=\frac{1}{2}\left(z+\frac{q_0^2}{z}\right),\quad \lambda=\frac{1}{2}\left(z-\frac{q_0^2}{z}\right).$

\begin{figure}[!t]
\centering
\includegraphics[scale=0.42]{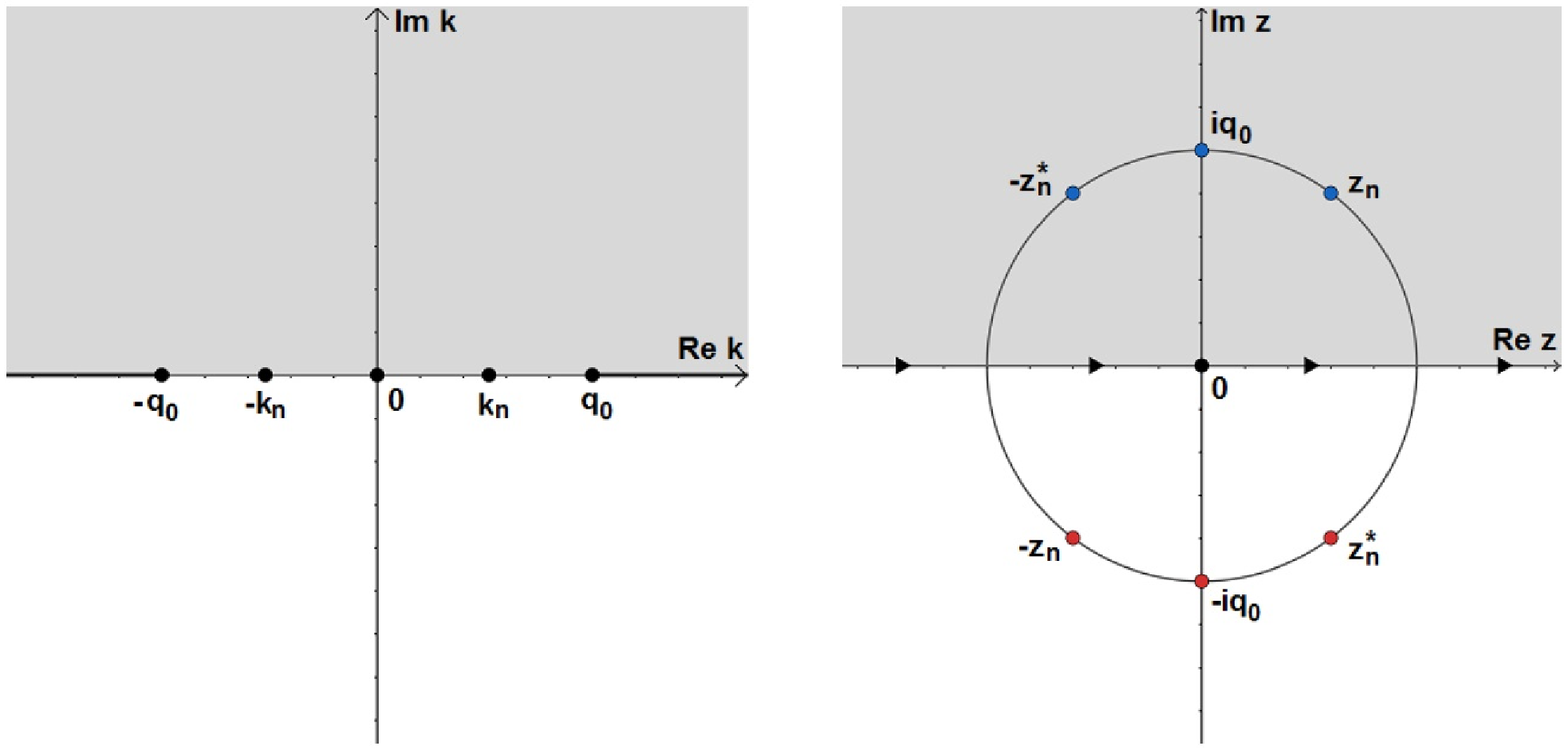}
\caption{Defousing mKdV equation with NZBCs. Left: the first sheet of the Riemann surface, showing the discrete spectrums, the region where $\mathrm{Im}\,\lambda>0$ (grey) and the region where $\mathrm{Im}\,\lambda<0$ (white). Right: the complex $z$-plane, showing the discrete spectrums [zeros of $s_{1,1}$ (blue) in grey region and those of $s_{2,2}$ (red) in white region], the region $\mathbb{C}^+$ where $\mathrm{Im}\,\lambda>0$ (grey), the region $\mathbb{C}^-$ where $\mathrm{Im}\,\lambda<0$ (white) and the orientation of the contours for the Riemann-Hilbert problem.}
\label{d}
\end{figure}

The mapping relations  between the Riemann surface and complex $z$-plane are exhibited as follows (see Fig.~\ref{d}):
\begin{itemize}
\item The Sheet-I and Sheet-II, excluding the branch cut, are mapped onto the exterior and interior of the circle of radius $q_0$, respectively;

\item The branch cut $\left[-q_0,q_0\right]$ are mapped onto the circle of radius $q_0$. In particular, the branch cut on Sheet-I (Sheet-II) is mapped onto the lower (upper) semicircle of radius $q_0$;

\item $\left(-\infty, -q_0\right]\cup \left[q_0, +\infty\right)$ is mapped onto the real $z$ axis. In particular, $\left[q_0,+\infty\right)$ of the Sheet-I (Sheet-II)  is mapped onto $\left[q_0,+\infty\right)$ ($\left[0, q_0\right]$) and $\left(-\infty,
-q_0\right]$ of the Sheet-I (Sheet-II)  is mapped onto $\left(-\infty,-q_0\right)$ ($\left[-q_0, 0\right]$);

\item The region where Im $\lambda>0$ (Im $\lambda<0$) of the Riemann surface is mapped onto the grey (white) domain in the complex $z$-plane. In particular, the UHP and LHP of the Sheet-I (Sheet-II) are mapped, respectively,  onto the UHP (LHP) and LHP (UHP) outside (inside) of the circle of radius $q_0$ in the complex $z$-plane.
\end{itemize}

The uniformization variable defines a map from the Riemann surface onto the complex plane, which will allow us to work with a complex parameter $z$ instead of dealing with the more cumbersome two-sheeted Riemann surface. For convenience, we denote the grey and white domain in Fig. \ref{d} (Right) by $D_+=\{z\in\mathbb{C}: {\rm Im}(\lambda(z))=\frac12
{\rm Im}(z)(1+q_0^2/|z|^2)>0\}$  and $D_-=\{z\in\mathbb{C}: {\rm Im}(\lambda(z))=\frac12
{\rm Im}(z)(1+q_0^2/|z|^2)<0\}$, respectively. In the following, we will consider our problem on the complex $z$-plane, in which one can rewrite the fundamental matrix solution of the asymptotic scattering problem as $\varPhi(x, t; z)=E_{\pm}(z)\exp{\left[i\,\theta(x, t, z)\,\sigma_3\right]}$,
where
\begin{align*}
E_{\pm}\left(z\right)=\begin{bmatrix}
1&\frac{iq_{\pm}}{z}\\-\frac{iq_{\pm}}{z}&1
\end{bmatrix},\quad
\theta\left(x, t; z\right)=\frac{1}{2}\left(z-\frac{q_0^2}{z}\right)\left\{x+\left[\left(z+\frac{q_0^2}{z}\right)^2+2q_0^2\right]t\right\},\quad z\not=\pm q_0.
\end{align*}

\begin{remark}
 For $z=\pm q_0$, we have ${\rm det}E_{\pm}(z)=0$, while when $z\not=\pm q_0$, ${\rm det}E_{\pm}(z)=1-q_0^2/z^2=\gamma_d(z)\not=0$, the inverse of $E_{\pm}(z)$ exists. Moreover, we find that $ X_{\pm}T_{\pm}= T_{\pm} X_{\pm}=(4k^2+2q_0^2)X_{\pm}^2=(2q_0^2+4k^2)(q_0^2-k^2)I$, and
  \bee\label{xtd}
  \begin{array}{l}
   X_{\pm}E_{\pm}(z)=\d i\lambda E_{\pm}(z)\sigma_3=\frac{i}{2}\left(z+\frac{q_0^2}{z}\right)E_{\pm}(z)\sigma_3,\\
   T_{\pm}E_{\pm}(z)=i\lambda (4k^2+2q_0^2) E_{\pm}(z)\sigma_3
   =\d\frac{i}{2}\left(z+\frac{q_0^2}{z}\right)\left[\left(z-\frac{q_0^2}{z}\right)^2+2q_0^2\right]E_{\pm}(z)\sigma_3,
   \end{array}
     \ene
which allows us to define the Jost solutions as simultaneous solutions of both parts of the Lax pair (\ref{lax-x}) and (\ref{lax-t}).
\end{remark}

\subsubsection{Jost solutions, analyticity, and continuity}

The continuous spectrum is  $\Sigma:=\mathbb{R}\backslash\{0\}$. As $z\in\Sigma$, we will seek for the Jost solutions $\varPhi_{\pm}(x, t; z)$ such that
\begin{align}\label{Jost-asy-1}
\varPhi_{\pm}(x, t; z)=E_{\pm}(z)\,\mathrm{e}^{i\,\theta(x, t; z)\,\sigma_3}+o\left(1\right),\quad x\to\pm\infty.
\end{align}
Factorizing the asymptotic exponential oscillations, we introduce the modified Jost solutions as
\begin{align}\label{bianjie-1}
\mu_{\pm}(x, t; z)=\varPhi_{\pm}(x, t; z)\,\mathrm{e}^{-i\theta(x, t; z)\sigma_3},
\end{align}
such that $\lim\limits_{x\to\pm\infty}\mu_{\pm}(x, t; z)=E_{\pm}(z)$. The Jost integral equation for $\mu_{\pm}(x, t; z)$ is posed as
\begin{align}\label{Jost-int-1}
\mu_{\pm}(x, t; z)=\left\{
\begin{aligned}
&E_{\pm}(z)+\int_{\pm\infty}^xE_{\pm}(z)\,\mathrm{e}^{i\lambda(z)(x-y)\widehat\sigma_3}\left[E^{-1}_{\pm}(z)\,\Delta Q_{\pm}(y, t)\,\mu_{\pm}(y, t; z)\right]\,\mathrm{d}y, && z\ne\pm q_0,\\[0.05in]
&E_{\pm}(z)+\int_{\pm\infty}^x\left[I+\left(x-y\right)\left(Q_{\pm}\pm i\,q_0\,\sigma_3\right)\right]\Delta Q_{\pm}(y, t)\,\mu_{\pm}(y, t; z)\,\mathrm{d}y, &&  z=\pm q_0,
\end{aligned}\right.
\end{align}

\begin{proposition}
Suppose $\left(1+|x|\right)\left(q-q_{\pm}\right)\in L^1\left(\mathbb{R^{\pm}}\right)$, then $\mu_\pm(x, t; z)$ and $\varPhi_\pm(x, t; z)$ have the following  properties:
\begin{itemize}
\item The Jost integral equation (\ref{Jost-int-1}) has unique solutions $\mu_{\pm}(x, t; z)$ in $\Sigma$.  The existence and uniqueness for $\varPhi_\pm(x, t; z)$ follows trivially.
\item $\mu_{+1}(x, t; z)$ and $\mu_{-2}(x, t; z)$ can be extended analytically to $\mathbb{C}^{+}$ and continuously to $\mathbb{C}^{+}\cup\Sigma$. $\mu_{-1}(x, t; z)$, and $\mu_{+2}(x, t; z)$ can be extended analytically to $\mathbb{C}^{-}$ and continuously to $\mathbb{C}^-\cup\Sigma$. The analyticity and continuity properties for $\varPhi_{\pm}(x, t; z)$ follow trivially.
\item The asymptotic behaviors for $\mu_{\pm}(x, t; z)$ as $z\to\infty$ and $z\to 0$ are
\bee \no
\mu_{\pm}(x, t; z)=\left\{\begin{array}{ll}
 I+O\left(\dfrac{1}{z}\right),& z\to\infty,\v\\
 \dfrac{i}{z}\,\sigma_3\,Q_{\pm}+O\left(1\right),& z\to 0.
 \end{array}\right.
\ene

\item $\varPhi_\pm(x, t; z)$ has the following three symmetries:
\begin{gather} \label{dhsi}
\varPhi_\pm(x, t; z)=\sigma_1\,\varPhi_\pm(x, t; z^*)^*\,\sigma_1, \quad
\varPhi_\pm(x, t; z)=\varPhi_\pm(x, t; -z^*)^*,\quad
\varPhi_\pm(x, t; z)=\left(-\frac{q_\pm}{z}\right)\,\varPhi_\pm\left(x, t; \frac{q_0^2}{z}\right)\sigma_2.
\end{gather}
\end{itemize}
\end{proposition}

\subsubsection{Scattering matrix and discrete spectrum}

Using the Liouville's formula, we can also define the scattering matrix (\ref{Jostchuandi}), scattering coefficients (\ref{S-lie}) and reflection coefficients (\ref{fanshe}). Next, we present three symmetries.
\begin{proposition}
The three symmetries for the scattering matrix are given by
\begin{itemize}
\item The first symmetry
\begin{gather} \label{s1}
S(z)=\sigma_1\,S(z^*)^*\,\sigma_1.
\end{gather}
\item The second symmetry
\begin{gather}\label{ss1}
S(z)=S(-z^*)^*.
\end{gather}
\item The third symmetry
\begin{gather}
S(z)=\frac{q_+}{q_-}\,\sigma_2\,S\left(\frac{q_0^2}{z}\right)\,\sigma_2.
\end{gather}
\end{itemize}
\end{proposition}
Using the first symmetry (\ref{s1}) and $\mathrm{det}\,S(z)=1$, one can find that there are no spectral singularities. In \cite{Faddeev1987}, it was shown that the discrete spectral points are simple. In \cite{Demontis2013}, it was shown that if $\left(1+\left|x\right|\right)^4\left(q-q_\pm\right)\in L^1\left(\mathbb{R^{\pm}}\right)$, then there is a finite number of discrete spectral points, all of which belong to $\left\{z: \left|z\right|=q_0\right\}$. Therefore, we give the discrete spectrum as
\begin{gather}\label{dpz}
Z=\left\{z_n, z_n^*, -z_n^*, -z_n\right\}_{n=1}^N\cup\left\{iq_0, -iq_0\right\}^\delta,
\end{gather}
where $z_n$  satisfies that $\left|z_n\right|=q_0, \mathrm{Re}\,z_n>0,  \mathrm{Im}\,z_n>0$, $\delta=0$ or $ 1$ and $\left\{iq_0, -iq_0\right\}^\delta$ is a set defined by
\begin{gather}
\left\{iq_0, -iq_0\right\}^\delta=\left\{
\begin{aligned}
&\qquad \varnothing, \quad &\delta&=0,  \\
&\left\{iq_0, -iq_0\right\}, \quad & \delta&=1,
\end{aligned}
\right.
\end{gather}
For convenience, using Eq.~(\ref{bAdingyi}) one obtains the same residue as Eq.~(\ref{Jie-liushu}).
\begin{proposition}
For $z_0\in Z$ defined by Eq.~(\ref{dpz}), three relations for the residue $A[z_0]$ are given by
\begin{itemize}
\item The first relation
$A[z_0]=A[z_0^*]^*.$
\item The second relation
$A[z_0]=-A[-z_0^*]^*.$
\item The third relation
$A[z_0]=\frac{z_0^2}{q_0^2}\,A\left[\frac{q_0^2}{z_0}\right].$
\end{itemize}
\end{proposition}
\begin{corollary}
For $n=1, 2, \cdots N$ and $\delta=1$, one has
\begin{gather}
\begin{gathered}
A[z_n]=A[z_n^*]^*=-A[-z_n^*]^*=-A[-z_n],\quad
\mathrm{Im}\, A[z_n]=\frac{1}{\mathrm{Im}\,\frac{q_0^2}{z_n^2}}\,\left[\left(\mathrm{Re}\,\frac{q_0^2}{z_n^2}-1\right)\mathrm{Re}\,A[z_n]\right],\\
A[iq_0]=A[-iq_0]^*, \quad \mathrm{Re}\,A[iq_0]=0.
\end{gathered}
\end{gather}
\end{corollary}
The asymptotic behaviors for the scattering matrix $S(z)$ as $z\to\infty$ and $z\to 0$ are needed to formulate the Riemann-Hilbert problem. We find that the asymptotic behaviors are same as Eqs.~(\ref{ien}, \ref{ien-1}).

\subsection{Inverse problem  with NZBCs}

\subsubsection{Riemann-Hilbert problem and reconstruction formula}

The Riemann-Hilbert problem for the defocusing mKdV equation with NZBCs also has the form Eqs.~(\ref{RHP-M}, \ref{RHP-Jump}, \ref{RHP-Asy}). To solve it, it is convenient to define that
\begin{gather}
\eta_n=\left\{
\begin{aligned}
&\quad z_n,  &&n=1, 2, \cdots, N,  \\
&-z_{n-N}^*,  &&n=N+1, N+2, \cdots, 2N,
\end{aligned}\right.
\qquad
\eta_{2N+\delta}=\left\{
\begin{aligned}
&\eta_{2N}, &  & \delta=0,  \\
&iq_0, &  & \delta=1.
\end{aligned}\right.
\end{gather}
Subtracting out the asymptotic and the pole contribution, the jump condition (\ref{RHP-Jump}) becomes
\begin{gather}
\begin{aligned}
&M^-(x, t; z)-I-\frac{i}{z}\,\sigma_3,Q_--\sum_{n=1}^{2N+\delta}\left[\frac{\mathop\mathrm{Res}\limits_{z=\eta_n}M^+(x, t; z)}{z-\eta_n}+\frac{\mathop\mathrm{Res}\limits_{z=\eta_n^*}M^-(x, t; z)}{z-\eta_n^*}\right]\\[0.05in]
&\quad =M^+(x, t; z)-I-\frac{i}{z}\,\sigma_3,Q_--\sum_{n=1}^{2N+\delta}\left[\frac{\mathop\mathrm{Res}\limits_{z=\eta_n}M^+(x, t; z)}{z-\eta_n}+\frac{\mathop\mathrm{Res}\limits_{z=\eta_n^*}M^-(x, t; z)}{z-\eta_n^*}\right]
-M^+(x, t; z)\,J(x, t; z).
\end{aligned}
\end{gather}

\begin{theorem}
The integral representation for the solution of the Riemann-Hilbert problem is
\begin{align}
M(x, t; z)=I+\frac{i}{z}\,\sigma_3\,Q_-+\sum_{n=1}^{2N+\delta}\left[\frac{\mathop\mathrm{Res}\limits_{z=\eta_n}M^+(z)}{z-\eta_n}+\frac{\mathop\mathrm{Res}\limits_{z=\eta_n^*}M^-(z)}{z-\eta_n^*}\right]+\frac{1}{2\pi i}\int_\Sigma\frac{\left(M^+J\right)(\zeta)}{\zeta-z}\,\mathrm{d}\zeta,\quad z\in\mathbb{C}\backslash\Sigma,
\end{align}
where $\int_\Sigma$ denotes the integral along the oriented contour shown in Fig. \ref{d}(right).
\end{theorem}

Next, we denote $M(x, t; z)$ as a closed system. Let
\begin{align}
C_n(z)=\frac{A[\eta_n]\,\mathrm{e}^{-2i\theta(x, t; \eta_n)}}{z-\eta_n}, \quad \widehat C_n(z)=\frac{A[\eta_n^*]\,\mathrm{e}^{2i\theta(x, t; \eta_n^*)}}{z-\eta_n^*}.
\end{align}
Then $M(x, t; z)$ can be written as
\begin{align}\label{ixbgs}
M(x, t; z)=I+\frac{i}{z}\,\sigma_3Q_-+\sum_{n=1}^{2N+\delta}\left[C_n(z)\,\mu_{-2}(\eta_n),\, \widehat C_n(z)\,\mu_{-1}(\eta_n^*)\right]+\frac{1}{2\pi i}\int_\Sigma\frac{\left(M^+J\right)(\zeta)}{\zeta-z}\,\mathrm{d}\zeta.
\end{align}
The remaining task is to evaluate $\mu_{-2}(\eta_n)$ and $\mu_{-1}(\eta_n^*)$. As $z\in D_+$, it follows from the second column of $M(x, t; z)$ in Eq.~(\ref{ixbgs}) that we obtain
\begin{align}\label{anei}
\mu_{-2}(z)=
\begin{bmatrix}
\frac{iq_-}{z}\\[0.05in]1
\end{bmatrix}
+\sum_{n=1}^{2N+\delta}\widehat C_n(z)\,\mu_{-1}(\eta_n^*)+\frac{1}{2\pi i}\int_\Sigma\frac{\left(M^+J\right)_2(\zeta)}{\zeta-z}\,\mathrm{d}\zeta.
\end{align}
The third symmetry for the Jost solutions (\ref{dhsi})  implies that
\begin{align}
\mu_{-2}(z)=\frac{iq_-}{z}\,\mu_{-1}\left(\frac{q_0^2}{z}\right).
\end{align}
With the aid of  $\eta_k^*=\frac{q_0^2}{\eta_k}$, let $z=\eta_k, k=1, 2, \cdots, N$ in Eq.~(\ref{anei}) and it becomes
\begin{align}\label{wdtn}
\begin{bmatrix}
\frac{iq_-}{\eta_k}\\[0.05in]1
\end{bmatrix}
+\sum_{n=1}^{2N+\delta}\left(\widehat C_n(\eta_k)-\frac{iq_-}{\eta_k}\,\delta_{k, n}\right)\mu_{-1}(\eta_n^*)+\frac{1}{2\pi i}\int_\Sigma\frac{\left(M^+J\right)_2(\zeta)}{\zeta-\eta_k}\,\mathrm{d}\zeta=0.
\end{align}
This is a linear system of $2N+\delta$  algebraic-integral equations with $2N+\delta$ unknowns $\mu_{-1}(\eta_n^*)$, by which $\mu_{-1}(\eta_n^*)$ can be determined. Then $\mu_{-2}(\eta_n)$ can be obtained by $\mu_{-2}(\eta_n)=\frac{iq_-}{\eta_n}\,\mu_{-1}\left(\eta_n^*\right)$. Substituting them to Eq.~(\ref{ixbgs}), the closed system for $M(x, t; z)$ is derived such that we have the asymptotic behavior for $M(x, t; z)$ as $z\to\infty$:
\begin{gather}
M(x, t; z)=I+\frac{1}{z}\,M^{(1)}(x, t; z)+O\left(\frac{1}{z^2}\right), \quad z\to\infty,
\end{gather}
where
\begin{align}
\begin{aligned}
M^{(1)}(x, t; z)=i\,\sigma_3\,Q_-+\sum_{n=1}^{2N+\delta}\left[A[\eta_n]\,\mathrm{e}^{-2i\theta(\eta_n)}\mu_{-2}(\eta_n), \,A[\eta_n^*]\,\mathrm{e}^{2i\theta(\eta_n^*)}\mu_{-1}(\eta_n^*)\right]-\frac{1}{2\pi i}\int_\Sigma \left(M^+J\right)(\zeta)\,\mathrm{d}\zeta,
\end{aligned}
\end{align}
which generates the following theorem:

\begin{theorem}
 The reconstruction formula for the potential of the defocusing mKdV equation with NZBCs is given by
\begin{align}\label{recons-d}
q(x, t)=q_--i\sum_{n=1}^{2N+\delta}A[\eta_n^*]\,\mathrm{e}^{2i\theta(\eta_n^*)}\mu_{-11}(\eta_n^*)+\frac{1}{2\pi}
\int_\Sigma\left(M^+J\right)_{12}(\zeta)\,\mathrm{d}\zeta.
\end{align}
\end{theorem}

\subsubsection{Trace formulae and theta condition}

In the same manner, one can give the trace formulae as
\begin{align}
s_{11}(z)&=\exp\left(-\frac{1}{2\pi i}\int_\Sigma\frac{\log\left[1-\rho(\zeta)\,\tilde\rho(\zeta)\right]}{\zeta-z}\,\mathrm{d}\zeta\right)\prod_{n=1}^{2N+\delta}\frac{z-\eta_n}{z-\eta_n^*},\\[0.05in]
s_{22}(z)&=\exp\left(\frac{1}{2\pi i}\int_\Sigma\frac{\log\left[1-\rho(\zeta)\,\tilde\rho(\zeta)\right]}{\zeta-z}\,\mathrm{d}\zeta\right)\prod_{n=1}^{2N+\delta}\frac{z-\eta_n^*}{z-\eta_n}.
\end{align}
Note that
\begin{gather}
\prod_{n=1}^{2N+\delta}\frac{z-\eta_n}{z-\eta_n^*}=\prod_{n=1}^N\frac{\left(z-z_n\right)\left(z+z_n^*\right)}{\left(z-z_n^*\right)\left(z+z_n\right)}\left(\frac{z-iq_0}{z+iq_0}\right)^\delta.
\end{gather}

Let $z\to 0$, one obtains the theta condition as
\begin{align}
\mathrm{arg}\,\frac{q_+}{q_-}=\frac{1}{2\pi}\int_\Sigma\frac{\log\left[1-\rho(\zeta)\,\tilde\rho(\zeta)\right]}{\zeta}\,\mathrm{d}\zeta+\delta\,\pi.
\end{align}
According to the symmetry in Eqs. (\ref{s1}, \ref{ss1}), one has $\rho(\zeta)\,\tilde\rho(\zeta)=\rho(-\zeta)\,\tilde\rho(-\zeta)$. Then
\begin{gather*}
\int_{-\infty}^0\frac{\log\left[1-\rho(\zeta)\,\tilde\rho(\zeta)\right]}{\zeta}\,\mathrm{d}\zeta=\int_{+\infty}^0\frac{\log\left[1-\rho(\zeta)\,\tilde\rho(\zeta)\right]}{\zeta}\,\mathrm{d}\zeta,
\end{gather*}
which infers that the theta condition is
$\mathrm{arg}\,\frac{q_+}{q_-}=\delta\,\pi,\, \delta=0, 1,$ which means that as $\delta=1$, one has the opposite boundary conditions $q_+=-q_-$ at infinity, whereas $\delta=0$, one has the same boundary conditions $q_+=q_-$ at infinity.

\begin{figure}[!t]
\centering
\includegraphics[scale=0.42]{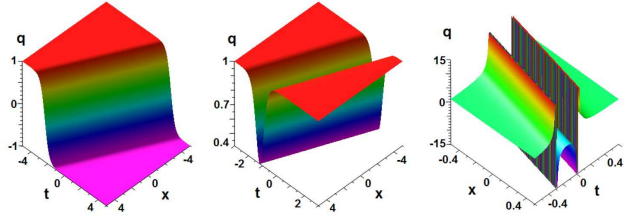}
\caption{Defousing mKdV equation with NZBCs. Left: kink solution with parameters: $N=0,\, \delta=1,\, q_{+}=-1,\, q_-=1,\, \mathrm{Im}\,A[i]=1$.  Middle: dark soliton with parameters: $N=1,\, \delta=0,\, q_-=1, \,z_1=\mathrm{e}^{\frac{\pi}{4}\,i}, \,\mathrm{Re}\,A[z_1]=1$. Right: singular solution with parameters: $N=1,\, \delta=0,\, q_-=1,\, z_1=\mathrm{e}^{\frac{\pi}{4}\,i},\, \mathrm{Re}\,A[z_1]=-\frac{3}{2}$.}
\label{d1}
\end{figure}

\subsubsection{Reflectionless potential with simple poles}

Eqs.~(\ref{recons-d}, \ref{wdtn}) yield the following theorem:
\begin{theorem}
The reflectionless potential of the defocusing mKdV equation with NZBCs (\ref{mKdV}) is given in terms of determinants
\begin{gather} \label{solu3}
q(x, t)=q_-+
\frac{\mathrm{det}
\begin{bmatrix}
H&\beta \vspace{0.06in} \\
\alpha^T&0
\end{bmatrix}}
{\mathrm{det}\,H}\,i,
\end{gather}
where $H=\left(h_{kn}\right)_{\left(2N+\delta\right)\times \left(2N+\delta\right)}$, $\alpha=\left(\alpha_n\right)_{\left(2N+\delta\right)\times 1}$, $\beta=\left(\beta_k\right)_{\left(2N+\delta\right)\times 1}$ wtih
\begin{gather}
h_{kn}=\widehat C_n(\eta_k)-\frac{iq_-}{\eta_k}\,\delta_{k, n},  \quad \alpha_n=A[\eta_n^*]\,\mathrm{e}^{2i\theta(x, t; \eta_n^*)}, \quad \beta_k=-\frac{iq_-}{\eta_k}
\end{gather}
\end{theorem}

\begin{figure}[!t]
\centering
\includegraphics[scale=0.42]{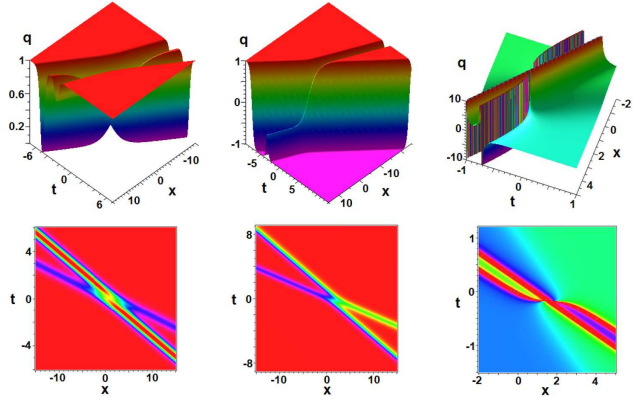}
\caption{Defousing mKdV equation with NZBCs. Left: $2$-dark-dark soliton with parameters: $N=2,\, \delta=0, \, q_-=1,\, z_1=\mathrm{e}^{\frac{\pi}{6}\,i},\, z_2=\mathrm{e}^{\frac{\pi}{3}\,i},\, \mathrm{Re}\,A[z_1]=\mathrm{Re}\,A[z_2]=1$. Middle: dark-soliton-kink solution with parameters $N=\delta=q_-=1,\, z_1=\mathrm{e}^{\frac{\pi}{4}\,i},\, \mathrm{Re}\,A[z_1]=\mathrm{Im}\,A[i]=1$. Right: interaction between the kink solution and singular solution with parameters $N=\delta=q_-=1,\, z_1=\mathrm{e}^{\frac{\pi}{4}\,i},\, \mathrm{Re}\, A[z_1]=-1, \, \mathrm{Im}\,A[i]=1$.}
\label{d2}
\end{figure}

The reflectionless potential contains parameters $N,\, \delta,\, q_-,\, z_n, \,\mathrm{Re}\,A[z_n],\, \mathrm{Im}\,A[iq_0],\, n=1, 2, \cdots, N$.
 The solution (\ref{solu3}) possesses the distinct wave structures for these different parameters:

\begin{itemize}

 \item {} As $\delta=0,\, N\not=0$, it exhibits an $N$-dark (or singular)  soliton solution of the defocusing mKdV equation with NZBCs $q_{+}=q_-$;

  \item {} As $\delta=1$, it stands for a kink-($N$-dark-soliton, or $N$-singular) solution of the defocusing mKdV equation with NZBCs $q_{+}=-q_-$;

   \end{itemize}

  In the following, we explicitly show some special wave structures for the double-pole soliton solution (\ref{solu3}) as follows:
 \begin{itemize}

 \item {} As $\delta=1,\, N=0$, the reflectionless one-kink soliton solution for the defocusing mKdV equation with NZBCs (\ref{mKdV}) reads
\begin{gather}
q(x,t)=q_--\frac{2cq_-\mathrm{e}^\varphi}{c\,\mathrm{e}^\varphi+2q_-},\quad c=\mathrm{Im}\,A[iq_0], \quad \varphi=2q_0\left(x+2q_0^2t\right),
\end{gather}
which is exhibited in Fig.~\ref{d1}(left) for the NZBCs $q_-=-q_+=1$;

 \item {} As $\delta=0,\, N=1$, the reflectionless $1$-soliton solution for the defocusing mKdV equation with NZBCs (\ref{mKdV}) is
\begin{gather}
q(x,t)=q_--\frac{8a\sin^2{c}\tan c\,q_0^2\,\mathrm{e}^\varphi}{\left(a\mathrm{e}^\varphi+2q_-\tan{c}\right)^2-4q_0^2\sin^2{c}\tan^2{c}},
\end{gather}
where
\begin{gather*}
c=\mathrm{arg}\,z_1\in\left(0,\frac{\pi}{2}\right),\quad a=\mathrm{Re}\,A[z_1],\quad \varphi=2q_0\sin{c}\left[x+2q_0^2\left(1+2\cos^2{c}\right)t\right].
\end{gather*}
As $aq_->0$, it displays a dark soliton (see Fig.~\ref{d1}(middle)), whereas $aq_-<0$, it is a singular solution (see Fig.~\ref{d1}(right)). This solution is blowup appearing at two lines of $x, t$, but the background of blowup solution is a non-zero constant, which means that this phenomenon may occur at the extreme conditions due to the effect of mutual repulsion of waves.

\begin{figure}[!t]
\centering
\includegraphics[scale=0.36]{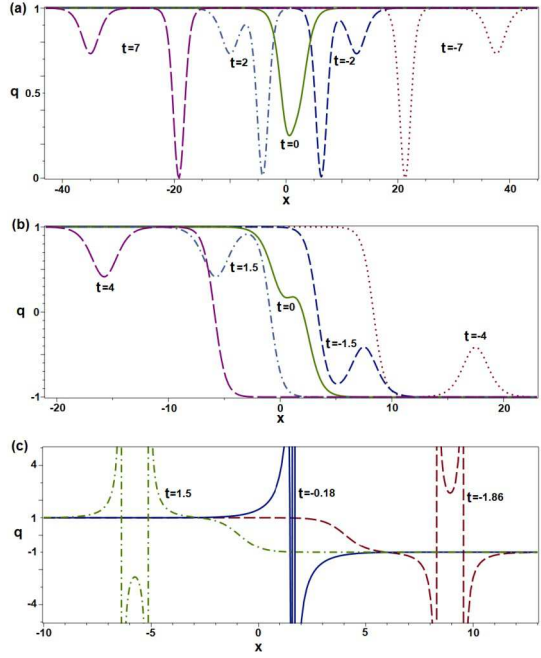}
\caption{Defousing mKdV equation with NZBCs. Profiles of the solutions given by Fig.~\ref{d2} for the distinct times. (a) Fig. \ref{d2} (left); (b) Fig. \ref{d2} (middle); (c) Fig. \ref{d2} (right).}
\label{d2p}
\end{figure}

  \item {} As $\delta=0,\, N=2$, Fig.~\ref{d2} (left) exhibits the elastic interaction of  two dark solitons with NZBCs $q_{\pm}=1$. Moreover, we  find that for $t=-7$, the low-amplitude dark soliton is located behind the high-amplitude one, after a short time, at near $t=0$ they consist of a dark soliton, and then the low-amplitude dark soliton is located before the high-amplitude one (see Fig.~\ref{d2p}(a)).

   \item {} As $\delta=1,\, N=1$, Fig.~\ref{d2} (middle) displays the interaction of one-dark soliton and one-kink soliton with NZBCs $q_-=-q_+=1$. Moreover, when $t=-4$, they display the interaction of a bright soliton and a kink soliton, and the background of the bright soliton is located on the right branch of the kink soliton (i.e., the kink approaches to $-1$ as $x\to \infty$). After a short time,
       the bright soliton becomes a dark soliton, and the dark soliton is located on the left branch of the kink soliton (i.e., the kink soliton approaches to $1$ as $x\to -\infty$) (see Fig.~\ref{d2p}(b)).

    \item {} As $\delta=1,\, N=1$, Fig.~\ref{d2} (right) illustrates  the interaction of one-kink soliton and one-singular solution with NZBCs $q_-=-q_+=1$. Moreover, when $t=-1.86$, they display the interaction of a singular solution and a kink soliton, and the singular points appear in the right branch of the kink soliton (i.e., the kink approaches to $-1$ as $x\to \infty$). After a short time,
       the singular points gradually appear in the left branch of the kink soliton (i.e., the kink soliton approaches to $1$ as $x\to -\infty$) (see Fig.~\ref{d2p}(c)).

  \end{itemize}

\section{Conclusions and discussions}

In conclusion, we have presented a rigorous theory of the IST for both focusing and defocusing mKdV equations with NZBCs such that we give their sing-pole and double-pole solutions by solving the corresponding Reimann-Hilbert problems. Moreover, we present the explicit expressions for the reflectionless potentials for the focusing and defocusing mKdV equations with NZBCs. Particularly, we exhibit the dynamical behaviors of some respective soliton structures including kink, dark, bright, breather solitons, and their interactions.

It should be pointed out that the difference between the mKdV equation with NZBCs and one with ZBCs is that the former needs to deal with a two-sheeted Riemann surface, which leads to different continuous  spectra, an additional symmetry, and more complicated discrete spectra. In contrast to the nonlinear Schr\"odinger equation with NZBCs (see,e.g., Refs.~\cite{Demontis2013,Biondini2014, Pichler2017}), the real potential for the mKdV equation with NZBCs in the ZS-AKNS scattering problem (Lax pair) adds more one reduction condition, which makes the corresponding direct and inverse problems be more complicated such as symmetries and discrete spectra. Moreover, the approach used in this paper can also be extended to study the ISTs for the mKdV equation with the asymmetric NZBCs, multi-component mKdV equations, nonlocal single or multi-component mKdV equations, and other nonlinear integrable systems with (asymmetric) NZBCs, which will be further discussed in other literatures.

\vspace{0.1in}
\baselineskip=15pt

\noindent {\bf Acknowledgements}

\vspace{0.05in}
The authors would like to thank Prof. G. Biondini for the valuable suggestions and discussions. This work was partially supported by the NSFC under grants Nos. 11731014 and 11571346, and CAS Interdisciplinary Innovation Team.

\end{document}